%% file: sample-aamas19.tex
\documentclass[sigconf]{aamas}  
\usepackage{balance}  

\usepackage{booktabs}

\setcopyright{ifaamas}  
\acmDOI{}  
\acmISBN{}  
\acmConference[AAMAS'19]{Proc.\@ of the 18th International Conference on Autonomous Agents and Multiagent Systems (AAMAS 2019)}{May 13--17, 2019}{Montreal, Canada}{N.~Agmon, M.~E.~Taylor, E.~Elkind, M.~Veloso (eds.)}  
\acmYear{2019}  
\copyrightyear{2019}  
\acmPrice{}  

\usepackage{amsmath}
\usepackage{amsfonts}
\usepackage{amssymb}
\usepackage{amsthm}
\usepackage{mathrsfs}
\usepackage{bm}
\usepackage{wrapfig}
\usepackage{pdfpages}
\newtheorem{Property}{Property}[section]
\newtheorem{Definition}{Definition}
\newtheorem{Proposition}{Proposition}

\newtheorem{theo}{theorem}

\newcommand{\ie}{\textit{i}.\textit{e}.}
\newcommand{\etal}{\textit{et al}.}
\usepackage[ruled]{algorithm2e}
\usepackage{subfigure}
\usepackage{graphicx}

\begin{document}

\title{Convergence of Multi-Agent Learning with a Finite Step Size in General-Sum Games}  



\author{Xinliang Song}
\affiliation{%
  \institution{Tsinghua University}
 \city{Beijing} 
 \country{China} 
}
\email{songxinliang@outlook.com}
\author{Tonghan Wang}
\affiliation{%
  \institution{Tsinghua University}
 \city{Beijing} 
 \country{China} 
}
\email{tonghanwang1996@gmail.com}
\author{Chongjie Zhang}
\affiliation{%
  \institution{Tsinghua University}
  \city{Beijing} 
  \country{China}}
\email{chongjie@tsinghua.edu.cn}

\begin{abstract}  
	Learning in a multi-agent system is challenging because agents are simultaneously learning and the environment is not stationary, undermining convergence guarantees. To address this challenge, this paper presents a new gradient-based learning algorithm, called Gradient Ascent with Shrinking Policy Prediction (GA-SPP), which augments the basic gradient ascent approach with the concept of shrinking policy prediction. The key idea behind this algorithm is that an agent adjusts its strategy in response to the forecasted strategy of the other agent, instead of its current one. GA-SPP is shown formally to have Nash convergence in larger settings than existing gradient-based multi-agent learning methods. Furthermore, unlike existing gradient-based methods, GA-SPP's theoretical guarantees do not assume the learning rate to be infinitesimal.
\end{abstract}

%


\maketitle


\input{samplebody-conf}


\bibliographystyle{ACM-Reference-Format}  
\balance  
\bibliography{sample-bibliography}  

\end{document}

%% file: samplebody-conf.tex
\input{1-Introduction}

\input{2-GA}
\input{3-GA-SPP}

\input{4-Analysis}

\input{5-1-Semi}
\input{5-2-Part-Anti}
\input{5-3-Two-Actions}

\input{6-Experiments}
\section{Conclusion}
This paper introduced a new gradient-based multi-agent learning algorithm, called gradient-ascent with shrinking policy prediction (GA-SPP). We proved Nash convergence of GA-SPP with a finite step size in three classes of general-sum games: $m\times n$ positive semi-definite games, a subclass of $2\times n$ general-sum games, and $2\times2$ general-sum games, respectively, which provide a stronger theoretical guarantee than existing gradient-based MAL algorithms. We also empirically verified the strong convergence property of GA-SPP with example games. In future work, we aim to relax assumptions of GA-SPP and extend it to stochastic games where each agent only has observations of their in-game payoffs and has no gradient information about other agents.


%% file: 1-Introduction.tex
\section{Introduction}
Multi-agent learning (MAL) is concerned with a set of agents that learn to maximize their expected rewards. There are a number of important applications that involve MAL, including competitive settings such as self-play in AlphaZero~\cite{silver2017mastering} and generative adversarial networks in deep learning~\cite{goodfellow2014generative, Metz2016Unrolled}, cooperative settings such as when learning to communicate~\cite{foerster2016learning,sukhbaatar2016learning} and multiplayer game~\cite{foerster2017counterfactual}, or some mix of the two~\cite{tampuu2017multiagent,leibo2017multi}. Although promising empirical results, establishing a theoretical guarantee of convergence for MAL, especially for gradient-based methods, is fundamentally challenging because of its non-stationary environment.

Recent multi-agent learning (MAL) algorithms~\cite{crandall2011learning,crandall2014towards,prasad2015two,bošanský2016algorithms,meir2017iterative,damer2017safely} with satisfactory empirical results are proposed, but most of them do not provide theoretical analyses of convergence. There are only a few worksthat provide theoretical results before them. Singh \etal~\cite{singh2000nash} first consider the theoretical convergence of gradient-based methods in MAL. After that, several variants~\cite{singh2000nash,bowling2005convergence,Banerjee2007Generalized,abdallah2008multiagent,zhang2010multi} are proposed and they provide theoretical convergence in general-sum games, but theoretical guarantees are restricted in 2-agent, 2-action games and they assume that the learning rate is infinitesimal, which is not practical. Some other online learning algorithms~\cite{daskalakis2011near,krichene2015online,cohen2017learning} have also been proposed with theoretical guarantees, but just for specific settings, such as congestion games and potential games.


In this paper, we propose a new multi-agent learning algorithm that augments a basic gradient ascent algorithm with \emph{shrinking} policy prediction, called Gradient Ascent with Shrinking Policy Prediction (GA-SPP). The key idea behind this algorithm is that an agent adjusts its strategy in response to the forecasted strategy of the other agent, instead of its current one. This paper makes three major novelties. First, to our best knowledge, GA-SPP is the first gradient-ascent MAL algorithm with a finite learning rate that provides convergence guarantee in general-sum games. Second, GA-SPP provides convergence guarantee in larger games than existing gradient-ascent MAL algorithms, which include $m\times n$ positive semi-definite games, a class of $2\times n$ general-sum games, and $2 \times 2$ general-sum games. Finally, GA-SPP guarantees to converge to a Nash Equilibrium when converging in any $m\times n$ general-sum game. 

Although GA-SPP shares some similar ideas about using policy prediction with IGA-PP~\cite{zhang2010multi} and the extra-gradient method~\cite{Antipin2003Extragradient}, it has several major differences from them. For example, apart from using a finite step size, another significant difference between GA-SPP and IGA-PP is that forecasted strategies of the opponent are projected to the valid probability space. This improvement enables GA-SPP's Nash convergence when converging, which does not hold for IGA-PP. In contrast to the extra-gradient approach, GA-SPP uses shrinking prediction lengths which can be different from the policy update rate. This improvement makes GA-SPP not only more flexible in practice but also stronger in terms of theoretical guarantees.

Like IGA-PP, we assume that agents know the other agent's strategy and its current strategy gradient, but we do not require the learning rate to be infinitesimal. Even though GA-SPP needs some restricted assumptions, it pushes forward the state of the art of MAL with theoretical analysis. We expect that our work can shed a light for theoretical understanding of dynamics and complexity of MAL problems and like IGA-PP and WoLF-IGA \cite{bowling2001convergence}, can encourage broadly applicable multi-agent reinforcement learning algorithms. Our proposed learning algorithm also provides a different approach for computing Nash Equilibiria of subsets of larger games, other than well-established offline algorithms~\cite{Lemke1964Equilibrium,Porter2004Simple}, whose computation complexity increases sharply with the number of actions.


\subsection*{Notation}
We use following notations in this paper:\\
\indent $\bm{\Delta}$~denotes the valid strategy space (\ie, a simplex).\\
\indent $\Pi_{\bm{\Delta}}$: $\Re^n\rightarrow \bm{\Delta}$~denotes the convex projection to the valid space,
\[ \Pi_{\bm{\Delta}} [\bm{x}] = argmin_{\bm{z}\in\bm{\Delta} }\| \bm{x}- \bm{z} \|. \]
\indent $P_{\bm{\Delta}}(\bm{x},\ \bm{v})$~denotes the projection of a vector $\bm{v}$~on $\bm{x}\in\bm{\Delta}$,
	\[ P_{\bm{\Delta}} (\bm{x},\ \bm{v}) =\lim_{\eta\rightarrow 0}\frac {\Pi_{\bm{\Delta}}[\bm{x}+\eta\bm{v}]-\bm{x}}{\eta}.\]
\indent $(\bm{v_1};\bm{v_2})$~denotes~${\left( \begin{array}{c}\bm{v_1}\\\bm{v_2}\end{array} \right)}$, where $\bm{v_1},\ \bm{v_2}$~are column vectors.

%% file: 2-GA.tex
\section{Gradient Ascent }
We begin with a brief overview of normal-form games and then review the basic gradient ascent algorithm.
\subsection{Normal-Form Games}
A 2-agent, $m\times n$ -action, general-sum normal-form game is defined by a pair of matrices 
\begin{gather*}
R=\begin{bmatrix} r_{11} & ... & r_{1n} \\... & ... & ...\\ r_{m1} & ... & r_{mn} \end{bmatrix}\quad \text{and}\quad
C=\begin{bmatrix} c_{11} & ... & c_{1n}  \\... & ... & ...\\ c_{m1} & ... & c_{mn} \end{bmatrix}
\end{gather*}
specifying the payoffs for the row agent and the column agent, respectively. The agents simultaneously select an action from their available set, and the joint action of the agents determines their payoffs according to their payoff matrices. If the row agent selects action $a\in \{ 1,\ ...,\ m\}$ and the column agent selects action $b\in \{ 1,\ ...,\ n\} $, respectively, then the row agent receives a payoff $r_{ab}$ and the column agent receives a payoff $c_{ab}$.

The agents can choose actions stochastically based on some probability distribution over their available actions. This distribution is said to be a mixed strategy. Let $\alpha_i \in [0,\ 1]$ denote the probability of choosing the i-th action by the row  agent and $\beta_j \in [0,\ 1]$ denote the probability of choosing the j-th action by the column agent, where $i \in  \{ 1,\ ...,\ m-1\},\ j \in  \{ 1,\ ...,\ n-1\}$, $\sum_1^{m-1} \alpha_i\leq 1$,  $\sum_1^{n-1} \beta_j\leq 1$. We use $\bm{\Delta_1}$ to denote a m-1 dimensional simplex and $\bm{\Delta_2}$ to denote a n-1 dimensional simplex. This $(m-1)\times(n-1)$ representation is equivalent to the $m\times n$ representation, and, following the previous work on gradient-based methods, we choose the former one. Let
\begin{equation}\begin{aligned}\nonumber
		\bm{\alpha}={[\alpha_1 \ ...\ \alpha_{m-1}]}^\mathrm{T},\qquad
		\bm{e_{m-1}}={[1 \ ...\ 1]}^\mathrm{T},\qquad \ \ \ \ \\
		\bm{\beta}={[\beta_1 \ ...\ \beta_{n-1}]}^\mathrm{T},\qquad\ \ 
		\bm{e_{n-1}}={[1 \ ...\ 1]}^\mathrm{T},\qquad \ \ \ \ \\
\end{aligned}\end{equation}
\noindent where the dimension of $\bm{e_{m-1}}$ is $m-1$, the dimension of $\bm{e_{n-1}}$ is $n-1$.

Then $\bm{\alpha}\in \bm{\Delta_1},\ \bm{\beta}\in \bm{\Delta_2}$.
With a joint strategy $(\bm{\alpha} ,\bm{\beta})$, the row agent's and column agent's expected payoffs are
\begin{equation}
\begin{aligned}
	\label{equa:VrVc}
	V_r(\bm{\alpha},\ &\bm{\beta})=(\bm{\alpha};\ 1-\bm{e_{m-1}}^\mathrm{T}\bm{\alpha})^\mathrm{T} \bm{R} (\bm{\beta};\ 1-\bm{e_{n-1}}^\mathrm{T}\bm{\beta}), \\
	V_c(\bm{\alpha},\ &\bm{\beta})=(\bm{\alpha};\ 1-\bm{e_{m-1}}^\mathrm{T}\bm{\alpha})^\mathrm{T} \bm{C} (\bm{\beta};\ 1-\bm{e_{n-1}}^\mathrm{T}\bm{\beta}).
\end{aligned}
\end{equation}

A joint strategy $(\bm{\alpha^*},\ \bm{\beta^*})$~is called a Nash equilibrium if for any mixed strategy $\bm{\alpha}$~of the row agent, $V_r(\bm{\alpha^*},\ \bm{\beta}^*)\geq V_r(\bm{\alpha},\ \bm{\beta}^*)$, and for any mixed strategy $\bm{\beta}$~of the column agent, $V_c(\bm{\alpha^*},\ \bm{\beta}^*)\geq V_c(\bm{\alpha^*},\ \bm{\beta})$.
It is well-known that every game has at least one Nash equilibrium.
\subsection{Learning using Gradient Ascent in Iterated Games}
In an iterated normal-form game, agents repeatedly play the same game. Each agent seeks to maximize its expected payoff in response to the strategy of the other agent. Using the basic gradient ascent algorithm, a agent can increase its expected payoff by updating its strategy with a step size along the gradient of the current strategy. The gradient is computed as the partial derivative of the agent's expected payoff with respect to its strategy:
\begin{flalign}\begin{aligned}
	\label{partialvrvc}
	&\bm{\partial_\alpha} V_r(\bm{\alpha},\ \bm{\beta})=\frac{\partial V_r(\bm{\alpha} ,\  \bm{\beta})}{\bm{\partial \alpha}}= (\bm{I_{m-1}}\ \bm{e_{m-1}})\bm{R} (\bm{\beta};\ 1-\bm{e_{n-1}}^\mathrm{T}\bm{\beta}), \\
	&\bm{\partial_{\beta}} V_c(\bm{\alpha},\ \bm{\beta})=\frac{\partial V_c(\bm{\alpha} ,\ \bm{\beta})}{\bm{\partial \bm\beta}}=(\bm{I_{n-1}}\ \bm{e_{n-1}})\bm{C}^\mathrm{T} (\bm{\alpha};\ 1-\bm{e_{m-1}}^\mathrm{T}\bm{\alpha}),
\end{aligned}\end{flalign}
where $\bm{I_{m-1}}$~is (m-1)-order identity matrix and $\bm{I_{n-1}}$~is (n-1)-order identity matrix.

If $(\bm{\alpha_k},\ \bm{\beta}_k)$~are the strategies on the $k$-th iteration and both agents use gradient ascent, then the new strategies will be:
\begin{equation}\begin{aligned}\label{eq:ga_update_rule}
	&\bm{\alpha_{k+1}}=\Pi_{\bm{\Delta_1}} [\bm{\alpha_k}+\eta \bm{\partial_\alpha} V_r(\bm{\alpha_k},\ \bm{\beta_k})], \\
	&\bm{\beta_{k+1}}=\Pi_{\bm{\Delta_2}} [\bm{\beta_k}+\eta \bm{\partial_{\beta}} V_c(\bm{\alpha_k},\ \bm{\beta_k})],
\end{aligned}\end{equation}
where $\eta$ is the gradient step size. If the updates move the strategies out of the valid probability space, the function $\Pi_{\bm{\Delta}}$~will project it back.

Singh~\etal~\cite{singh2000nash} analyzed the gradient ascent algorithm by examining the dynamics of the strategies in the case of an infinitesimal step size $(\lim_{\eta \to 0})$. This algorithm is called Infinitesimal Gradient Ascent (IGA). IGA cannot converge in some 2-agent 2-action zero-sum game. GIGA-WoLF and IGA-PP extended IGA and provide theoretical guarantee of Nash equilibrium in 2-agent 2-action game through similar methods. However, these algorithms require a infinitesimal step size, which is not practical. We will describe a new gradient ascent algorithm that enables the agents' strategies to converge to a Nash equilibrium with a finite step size in a larger game setting.

%% file: 3-GA-SPP.tex
\section{Gradient Ascent With Shrinking Policy Prediction (GA-SPP)}
As shown in Eq.~\ref{eq:ga_update_rule}, the gradient used by IGA to adjust the strategy is based on current strategies. Suppose that an agent can estimate the change direction of the opponent's strategy, \ie, its strategy derivative, in addition to its current strategy. Then the agent can forecast the opponent's strategy and adjust its own strategy in response to the forecasted strategy. With this idea, we design a gradient ascent algorithm with shrinking policy prediction (GA-SPP). Its updating rule consists of three steps.

In Step 1, the new derivative terms with $\gamma$ serve as a short-term prediction of the opponent's strategy. If the opponent's forecasted strategy is out of boundary of simplex, it will be projected back to the valid space. 

In Step 2, agents update their strategies on the basis of the forecasted strategy of its opponent.

In Step 3, agents terminate or adjust their prediction lengths. If predicted strategies are equal to the current strategies, the algorithm will terminate. Step 3 can make sure GA-SPP only converges to Nash equilibrium (NE) instead of other points. Because when $(\bm{\alpha_{k+1}},\ \bm{\beta_{k+1}})=(\bm{\alpha_{k}},\ \bm{\beta_{k}})$, GA-SPP will stop, if there is no Step 3, then~$(\bm{\overline{\alpha}_{k+1}},\ \bm{\overline{\beta}_{k+1}})\ne(\bm{\alpha_{k}},\ \bm{\beta_{k}})$~may happen. In this situation, GA-SPP may converge to a non-NE point. We will prove this property of GA-SPP in Proposition 1.

\begin{algorithm}[h]
\caption{Updating rule of GA-SPP}\label{alg:gaspp}
1 $\bm{\overline{\alpha}_{k+1}}=\Pi_{\bm{\Delta_1}}[\bm{\alpha_k}+\gamma_k \bm{\partial_\alpha} V_r(\bm{\alpha_k},\ \bm{\beta_k})]$\;
$\ \ \ \bm{\overline{\beta}_{k+1}}= \Pi_{\bm{\Delta_2}}[\bm{\beta_k}+\gamma_k \bm{\partial_{\beta}} V_c(\bm{\alpha_k},\ \bm{\beta_k})]$\;

2 $\bm{\alpha_{k+1}}=\Pi_{\bm{\Delta_1}} [\bm{\alpha_k}+\eta \bm{\partial_\alpha} V_r(\bm{\alpha_{k}},\ \overline{\bm{\beta}}_{k+1})]$\; 
$\ \ \ \bm{\beta_{k+1}}=\Pi_{\bm{\Delta_2}} [\bm{\beta_k}+\eta \bm{\partial_{\beta}} V_c(\bm{\overline{\alpha}_{k+1}},\ \bm{\beta}_{k})];$\;

3

\eIf{$(\bm{\overline{\alpha}_{k+1}},\ \bm{\overline{\beta}_{k+1}})==(\bm{\alpha_k},\ \bm{\beta_k})$}{
    terminate\;
}
{
    \eIf{$(\bm{\alpha_{k+1}},\bm{\beta_{k+1}})==(\bm{\alpha_k},\bm{\beta_k})\ \&\ (\bm{\overline{\alpha}_{k+1}},\bm{\overline{\beta}_{k+1}})\ne(\bm{\alpha_k},\bm{\beta_k})$}{
        $\gamma_{k+1}=\mu_k\gamma_k,\ (0<\mu_k<1)$, back to (1)\;
    }
    {
        $\gamma_{k+1}=\gamma_k$, back to (1)\;
    }
}
\end{algorithm}

The prediction length $\gamma_k$ and gradient step size $\eta$ will affect the convergence of the GA-SPP algorithm. With a too large prediction length, the gradient computed with the forecasted strategy will deviate too much from the gradient computed with the opponent's current strategy. As a result, the agent may adjust its strategy in the improper direction and cause their strategies to fail to converge.

Following conditions ensure that $\gamma$ and $\eta$ are appropriate:\\
\indent\textbf{Condition 1:} \indent $\gamma_0>0$, $\eta>0$\\
\indent\textbf{Condition 2:} \indent $4\gamma_0^2\delta_r\delta_c<1$\\
\indent\textbf{Condition 3:} \indent $\eta,\ \gamma_0<\frac{1}{\delta_r+\delta_c}$\\
where $\delta_r=r_{max}-r_{min},\ \delta_c=c_{max}-c_{min}$,\ $r_{max}$ and $c_{max}$ is the maximum reward for the row and column agent,~$r_{min}$ and $c_{min}$ is the minimum reward for the row and column agent.

Condition 3 makes sure that the theoretical guarantee of Nash convergence in the game settings analyzed in Section~\ref{sec:convergence_of_ga-spp}. In experiments, the algorithm can still work in some other games if we choose larger prediction length or let agents have different prediction lengths. 


%% file: 4-Analysis.tex
\subsection{Analysis of GA-SPP}
In this section, we will show that if agents' strategies converge by following GA-SPP, then they must converge to a Nash equilibrium, which is described by Proposition~\ref{prop1}. Using this proposition, we will then prove the Nash convergence of GA-SPP in three classes of games: $m\times n$ positive semi-definite games, a class of $2\times n$ general-sum games, and $2\times 2$ general-sum games, respectively, in the following sections.
 
Before proving Proposition~\ref{prop1}, we will first show that if the projected gradients of a strategy pair are zero, then this strategy must be a Nash equilibrium, which is described by Lemma~\ref{lemma1}. For brevity, let $\bm{\partial_{\alpha}}$ denotes $\bm{\partial_{\alpha}}V_r(\bm{\alpha},\ \bm{\beta})$, $\bm{\partial_{\beta}}$ denotes $\bm{\partial_{\beta}}V_r(\bm{\alpha},\ \bm{\beta})$.

\begin{lemma}
	\label{lemma1}
	In ($m\times n$)-action games, if the projected partial derivatives at a strategy pair $(\bm{\alpha^*},\ \bm{\beta^*})$ are zero, that is, $P_{\bm{\Delta_1}}(\bm{\alpha^*},\ \bm{\partial_{\alpha^*}}) = 0$ and $P_{\bm{\Delta_2}}(\bm{\beta^*},\ \bm{\partial_{\beta^*}}) =0$, then $(\bm{\alpha^*},\ \bm{\beta^*})$ is a Nash equilibrium. 
\end{lemma}
\begin{proof}
	Assume that $(\bm{\alpha^*},\ \bm{\beta^*})$ is not a Nash equilibrium. Then at least one agent, for example, the column agent, can increase its expected payoff by changing its strategy unilaterally. Assume that the improved point is $(\bm{\alpha^*},\ \bm{\beta})$. Because of the convexity of the strategy space $\bm{\Delta_2}$ and the linear dependence of $V_c(\bm{\alpha},\ \bm{\beta})$ on $\bm{\beta}$, then, for any $\epsilon > 0$, $(\bm{\alpha^*},\ (1- \epsilon )\bm{\beta^*} + \epsilon\bm{\beta})$ must also be an improved point, which implies that the projected gradient of $\bm{\beta}$ at $(\bm{\alpha^*},\ \bm{\beta^*})$ is not zero. By contradiction, $(\bm{\alpha^*},\ \bm{\beta^*})$ is a Nash equilibrium.
\end{proof}

\begin{Proposition}
	\label{prop1}
		In 2-agent, $m\times n$ games, if two agents follow GA-SPP with appropriate $\gamma,\ \eta$ (satisfying Condition 1, 2, and 3) and GA-SPP converges, then $(\bm{\alpha^*},\ \bm{\beta^*})$ is a Nash equilibrium.
\end{Proposition}


Here is a proof sketch (the detailed formal proof is described in supplementary material\footnote{\label{supplementary_material}\url{https://drive.google.com/file/d/1TZeRf0xp4g4wg-JX7zA9TjqC2S619pAp/view?usp=sharing}}). According to Step 3 in the algorithm \ref{alg:gaspp}, if the strategy pair trajectory converges at \footnotesize$(\bm{\alpha^*},\ \bm{\beta^*})$\normalsize, then \footnotesize{}$(\bm{\alpha^*},\ \bm{\beta^*})=(\bm{\overline{\alpha}_{k+1}},\ \bm{\overline{\beta}_{k+1}})=(\bm{\alpha_k},\ \bm{\beta}_{k})$ \normalsize or \footnotesize$(\bm{\alpha^*},\ \bm{\beta^*})=\lim_{k\rightarrow\infty}(\bm{\overline{\alpha}_{k+1}},\ \bm{\overline{\beta}_{k+1}})=\lim_{k\rightarrow\infty}(\bm{\alpha_k},\ \bm{\beta}_{k})$\normalsize. For both cases, we can have \footnotesize$\bm{\alpha^*} = \Pi_{\Delta_1}[\bm{\alpha^*}+\eta \bm{\partial_{{\alpha^*}}}]$ \normalsize and \footnotesize$\bm{\beta^*} = \Pi_{\Delta_2}[\bm{\beta^*}+\eta \bm{\partial_{{\beta^*}}}]$\normalsize. 
From here, we can show that, for any arbitrary small \footnotesize$\epsilon>0$,\ $\bm{\alpha^*} = \Pi_{\Delta_1}[\bm{\alpha^*}+\epsilon \bm{\partial_{{\alpha^*}}}]$ \normalsize and \footnotesize$\bm{\beta^*} = \Pi_{\Delta_2}[\bm{\beta^*}+\epsilon \bm{\partial_{{\beta^*}}}]$\normalsize, which imply \footnotesize$P_{\bm{\Delta_1}}(\bm{\alpha^*},\  \bm{\partial_{\alpha^*}}) = \bm{0}$ \normalsize and \footnotesize$P_{\bm{\Delta_2}}(\bm{\beta^*},\ \bm{\partial_{\beta^*}}) =\bm{0}$\normalsize. Then according to Lemma~\ref{lemma1}, \footnotesize$(\bm{\alpha^*},\ \bm{\beta^*})$ \normalsize is a Nash equilibrium.

%% file: 5-1-Semi.tex
\section{Convergence of GA-SPP}\label{sec:convergence_of_ga-spp}
We will show the Nash convergence of GA-SPP in three classes of games in this section.

\subsection{$m\times n$ Positive Semi-Definite Games}

A function $\Phi(\bm{v},\ \bm{w})$ is called a positive semi-definite function if it obeys the inequality defined in~\cite{antipin1995convergence}:
\begin{equation}
    \ \Phi(\bm{w},\ \bm{w})-\Phi(\bm{w},\ \bm{v})-\Phi(\bm{v},\ \bm{w})+\Phi(\bm{v},\ \bm{v})\geq0.
\end{equation}

To facilitate the proof, we define the normalized value function for a game: 
\begin{equation}
    \label{Phi}
    \Phi(\bm{v},\ \bm{w})=V_r(\bm{\alpha^1},\ \bm{\beta^2})+V_c(\bm{\alpha^2},\ \bm{\beta^1}),
\end{equation}
where $\bm{v}=(\bm{\alpha^1},\ \bm{\beta^1})\in \{\bm{\Delta_1}\times\bm{\Delta_2}\}$, $\bm{w}=(\bm{\alpha^2},\ \bm{\beta^2})\in \{\bm{\Delta_1}\times\bm{\Delta_2}\}$.

\begin{Definition}\label{defsemi0}
	A 2-agent $m\times n$ game is called positive semi-definite (PSD) game if its normalized value function obeys
    \begin{equation}
    \ \Phi(\bm{w},\ \bm{w})-\Phi(\bm{w},\ \bm{v})-\Phi(\bm{v},\ \bm{w})+\Phi(\bm{v},\ \bm{v})\geq0.
\end{equation}
\end{Definition}

It means that for a PSD game, its payoff matrices satisfies
\begin{equation}
    \label{defsemi}
    \begin{aligned}
	&V_r(\bm{\alpha^1},\ \bm{\beta^1})+V_c(\bm{\alpha^1},\ \bm{\beta^1})+V_r(\bm{\alpha^2},\ 
	\bm{\beta^2})+V_c(\bm{\alpha^2},\ \bm{\beta^2}) \\
	\geq&V_r(\bm{\alpha^1},\ \bm{\beta^2})+V_c(\bm{\alpha^1},\ \bm{\beta^2})+V_r(\bm{\alpha^2},\  \bm{\beta^1})+V_c(\bm{\alpha^2},\ \bm{\beta^1})\\
	& \ \ \ \forall \bm{\alpha^1},\ \bm{\alpha^2}\in\bm{\Delta_{1}},\ \ \ \forall \bm{\beta^1},\ \bm{\beta^2}\in\bm{\Delta_{2}}.
	\end{aligned}
\end{equation}

Zero-sum games are a subset of PSD games, because their value functions satisfy $V_r(\bm{\alpha},\ \bm{\beta})+V_c(\bm{\alpha},\ \bm{\beta})=0$, then both sides of inequality~\ref{defsemi} are equal to zero.

For a PSD game,  if $(\bm{\alpha^*},\ \bm{\beta^*})$ is a Nash equilibrium and  $\bm{v^*}=(\bm{\alpha^*},\ \bm{\beta^*})$, then its normalized function obeys
    \begin{eqnarray}
    \label{semifunc}
    \langle\nabla_2\Phi(\bm{w},\ \bm{w}),\ \bm{w}-\bm{v^*}\rangle\geq 0\ \ \ \forall\bm{w} \in \{\bm{\Delta_1}\times\bm{\Delta_2}\}.
    \end{eqnarray}

    In the proof of Theorem~\ref{thesemi}, we will use this inequality.

\begin{theo}
	\label{thesemi}
	If, in a 2-agent, $m \times n$ iterated positive semi-definite norm-form game, both agents follow the GA-SPP algorithm (with  Condition 1, 2, and 3), then their strategies will converge to a Nash equilibrium. 
\end{theo}

\begin{proof}
	Motivated by~\cite{Antipin2003Extragradient}, our proof will use some variational inequalities techniques.
	 
	From the first and second step of GA-SPP (Algorithm~\ref{alg:gaspp}), we have estimates
\begin{equation}\label{est}
	\begin{aligned}
		&|\bm{\overline{\alpha}_{k+1}}-\bm{\alpha_{k+1}}|\leq |\gamma_k \bm{\partial_{\alpha}}V_r(\bm{\alpha_k},\ \bm{\beta_{k}})-\eta \bm{\partial_{\alpha}}V_r(\bm{\alpha_k},\ \bm{\overline{\beta}_{k+1}})|,\\
		&|\bm{\overline{\beta}_{k+1}}-\bm{\beta_{k+1}}|\leq |\gamma_k \bm{\partial_{\beta}}V_c(\bm{\alpha_{k}},\ \bm{\beta_{k}})-\eta \bm{\partial_{\beta}}V_c(\bm{\overline{\alpha}_{k+1}},\ \bm{\beta_{k}})|.
	\end{aligned}
\end{equation}

We present the first and second step of GA-SPP in the form of variational inequalities:
\begin{equation}\label{ve1}
    \begin{aligned}
	&\langle\bm{\overline{\alpha}_{k+1}}-\bm{\alpha_{k}}-\gamma_k \bm{\partial_{\alpha}}V_r(\bm{\alpha_k},\ 
	 \bm{\beta_k}),\ \bm{z_1}-\bm{\overline{\alpha}_{k+1}}\rangle\geq0\ \ \ \forall\bm{z_1}\in\bm{\Delta_1},\\
	&\langle\bm{\overline{\beta}_{k+1}}-\bm{\beta_{k}}-\gamma_k \bm{\partial_{\beta}}V_c(\bm{\alpha_k},\ \bm{\beta_k}),\ \bm{z_2}-\bm{\overline{\beta}_{k+1}}\rangle\geq0\ \ \ \forall\bm{z_2}\in\bm{\Delta_2};
    \end{aligned}
\end{equation}
\begin{equation}\label{ve2}
    \begin{aligned}
	&\langle\bm{\alpha_{k+1}}-\bm{\alpha_k}-\eta \bm{\partial_{\alpha}}V_r(\bm{\alpha_k},\ \bm{\overline{\beta}_{k+1}}),\ 
	\bm{z_1}-\bm{\alpha_{k+1}}\rangle \geq0\ \ \ \forall\bm{z_1}\in\bm{\Delta_1},\\
	&\langle\bm{\beta_{k+1}}-\bm{\beta_{k}}-\eta \bm{\partial_{\beta}}V_c(\bm{\overline{\alpha}_{k+1}},\ \bm{\beta_k}),\ \bm{z_2}-\bm{\beta_{k+1}}\rangle\geq0\ \ \ \forall\bm{z_2}\in\bm{\Delta_2}.
    \end{aligned}
\end{equation}
    
Let $\bm{v}$ = ${\left(\begin{array}{c}\bm{\alpha^1}\\\bm{\beta^1}\end{array}\right)}$. 
Put $\bm{z_1}=\bm{\alpha^*},\bm{z_2}=\bm{\beta^*}$ in Eq.~\ref{ve2}, then set $\bm{z_1}=\bm{\alpha_{k+1}},\bm{z_2}=\bm{\beta_{k+1}}$ in Eq.~\ref{ve1}, and take into account of Eq.~\ref{est}, we can get (the detailed computation is listed in our supplementary material)
\begin{equation}
   \begin{aligned}
   \label{ve8}
     &\langle\bm{v_{k+1}}-\bm{v_k},\ \bm{v^*}-\bm{v_{k+1}}\rangle+   \langle\bm{\overline{v}_{k+1}}-\bm{v_k},\ \bm{v_{k+1}}-\bm{\overline{v}_{k+1}}\rangle\\
   + & \eta\langle\nabla_2\Phi(\bm{\overline{v}_{k+1}},\ \bm{\overline{v}_{k+1}}),\ \bm{v^*}-\bm{\overline{v}_{k+1}}\rangle\\
   + & h^2\|\nabla_2\Phi(\bm{v_{k}},\ \bm{v_k})-\nabla_2\Phi(\bm{\overline{v}_{k+1}},\ \bm{\overline{v}_{k+1}})\|^2\ge0,
   \end{aligned}
\end{equation}
where $h=max\{\gamma_0, \eta\}$. By means of identity, the first two scalar products in Eq.~\ref{ve8} can be rewritten as
\begin{equation}\label{ve9_1}
   \begin{aligned} 
   \frac{1}{2}\|\bm{v_{k}}-\bm{v^*}\|^2-\frac{1}{2}\|\bm{v_{k+1}}-\bm{v^*}\|^2-\\
   \frac{1}{2}\|\bm{v_{k+1}}-\bm{\overline{v}_{k+1}}\|^2-\frac{1}{2}\|\bm{\overline{v}_{k+1}}-\bm{v_k}\|^2.
   \end{aligned}
\end{equation}
\indent Set $\bm{w}=\bm{\overline{v}_{k+1}}$ in Eq.~\ref{semifunc}, then the third term in Eq.\ref{ve8} is non-positive. For the last term of Eq.~\ref{ve8}, if $\nabla_2\Phi(\bm{v_{k}},\ \bm{v_k})$  satisfies the Lipschitz condition with constant $L$, then following estimate is correct 
\begin{equation}
   \label{ve9_3}
   |\nabla_2\Phi(\bm{v_{k}},\ \bm{v_k})-\nabla_2\Phi(\bm{\overline{v}_{k+1}},\ \bm{\overline{v}_{k+1}})|\leq L |\bm{\overline{v}_{k+1}}-\bm{v_k}|.
\end{equation}
   Now put Eq.~\ref{ve9_1} and Eq.~\ref{ve9_3} in Eq.~\ref{ve8}, we can yield
\begin{equation}\label{ve10}
   \begin{aligned}
   &\|\bm{v_{k+1}}-\bm{v^*}\|^2+
   \|\bm{v_{k+1}}-\bm{\overline{v}_{k+1}}\|^2+\\
   &(1-2h^2L^2)\|\bm{\overline{v}_{k+1}}-\bm{v_k}\|^2\leq\|\bm{v_{k}}-\bm{v^*}\|^2.
   \end{aligned}
\end{equation}
\indent Note that $\nabla_2\Phi(\bm{v_{k}},\ \bm{v_k})=\bm{\partial_{\alpha}}V_r(\bm{\alpha},\ \bm{\beta})+\bm{\partial_{\beta}}V_c(\bm{\alpha},\ \bm{\beta})$. According to Eq.~\ref{partialvrvc}, $\bm{\partial_{\alpha}}V_r(\bm{\alpha},\ \bm{\beta})$ is a function of $\bm{\beta}$, $\bm{\partial_{\beta}}V_c(\bm{\alpha},\ \bm{\beta})$ is a function of $\bm{\alpha}$. The maximum value of 2-norm of $\bm{\partial_{\alpha}}$ is not greater than ${\delta_r}^2/2$, and not greater than ${\delta_c}^2/2$ for 2-norm of $\bm{\partial_{\beta}}$. So the Lipschitz constant $L\leq\frac{\delta_r}{\sqrt{2}}+\frac{\delta_c}{\sqrt{2}}$.
According to Condition 3, $h= max \{\gamma_0,\ \eta \} < \frac{1} {\delta_c+\delta_r}$, so $h L<\frac{\sqrt{2}}{2}$ and $1-2h^2L^2>0$.
Sum up inequality Eq.~\ref{ve10} from $k=0$ to $k=K$, we get
\begin{equation}\label{ve11}
   \begin{aligned}
   &\|\bm{v_{K+1}}-\bm{v^*}\|^2+
   \sum_{k=0}^{K}\|\bm{v_{k+1}}-\bm{\overline{v}_{k+1}}\|^2+\\
   &(1-2h^2L^2)\sum_{k=0}^{K}\|\bm{\overline{v}_{k+1}}-\bm{v_k}\|^2\leq\|\bm{v_{0}}-\bm{v^*}\|^2.
   \end{aligned}
\end{equation}
\indent From the gained inequality (Eq.~\ref{ve11}) the bound of trajectory follows
\begin{equation}
   \label{ve12}
   \|\bm{v_{K+1}}-\bm{v^*}\|^2\leq\|\bm{v_{0}}-\bm{v^*}\|^2,
\end{equation}
and the series are convergent
\begin{equation}\nonumber
\sum_{k=0}^{K}\|\bm{v_{k+1}}-\bm{\overline{v}_{k+1}}\|^2<\infty,\ \sum_{k=0}^{K}\|\bm{\overline{v}_{k+1}}-\bm{v_k}\|^2<\infty.
\end{equation}
\indent As a result, 
$\lim_{k\rightarrow\infty}\|\bm{v_{k+1}}-\bm{\overline{v}_{k+1}}\|^2=0,\ \lim_{k\rightarrow\infty}\|\bm{\overline{v}_{k+1}}-\bm{v_k}\|^2=0$, so $\lim_{k\rightarrow\infty}\|\bm{v_{k+1}}-\bm{v_{k}}\|^2=0$. It implies $\lim_{k\rightarrow\infty}\|\bm{\alpha_{k+1}}-\bm{\alpha_{k}}\|^2=0$ and $\lim_{k\rightarrow\infty}\|\bm{\beta_{k+1}}-\bm{\beta_{k}}\|^2=0$.
   
So GA-SPP can converge. With Proposition~\ref{prop1}, GA-SPP must converge to a Nash equilibrium. Therefore, proof of Theorem~\ref{thesemi} is completed. 
\end{proof}

\begin{theo}
	\label{thesemi1}
	If, in a 2-agent, $m \times n$ iterated positive semi-definite norm-form game, one agent follows the GA-SPP algorithm (with  Condition 1, 2, and 3), another agent uses GA, then their strategies will converge to a Nash equilibrium. 
\end{theo}
    The proof of this theorem is omitted, which is similar to that of Theorem \ref{thesemi}.

%% file: 5-2-Part-Anti.tex
\subsection{A Subclass of $2\times n$ General-Sum Games}
In this section, we will show that GA-SPP converges to a Nash equilibrium in a subclass of 2-agent $2\times n$ general games (Theorem~\ref{thepart}).

A 2-agent, $2\times n$, general-sum normal-form game's payoff matrices can be written as
\begin{gather*}
R=\begin{bmatrix} r_{11} & ... & r_{1n} \\ r_{21} & ... & r_{2n} \end{bmatrix},\quad
C=\begin{bmatrix} c_{11} & ... & c_{1n} \\ c_{21} & ... & c_{2n} \end{bmatrix}.
\end{gather*}

Let
\begin{center}
	$\bm{r_1}={[r_{11} \ ...\ r_{1,n-1}]}^\mathrm{T}$,\ \ \ \ 
	$\bm{r_2}={[r_{21} \ ...\ r_{2,n-1}]}^\mathrm{T}$,\\
	$\bm{c_1}={[c_{11} \ ...\ c_{1,n-1}]}^\mathrm{T}$,\ \ \ \
	$\bm{c_2}={[c_{21} \ ...\ c_{2,n-1}]}^\mathrm{T}$.
\end{center}

Then agents' expected payoffs (Eq.~\ref{equa:VrVc}) are
\begin{equation}\label{pVrVc}
	\begin{aligned}
	V_r(\alpha,\ \bm{\beta})&=(\alpha \bm{\beta^\mathrm{T}})\bm{r_1}
	+r_{1n}(\alpha(1-\bm{\beta^\mathrm{T}}\bm{e_{n-1}}))  \\
	&+(1-\alpha) \bm{\beta^\mathrm{T}}\bm{r_2} +r_{2n}((1-\alpha)(1-\bm{\beta^\mathrm{T}}\bm{e_{n-1}})), \\
	V_c(\alpha,\ \bm{\beta})&=(\alpha \bm{\beta^T})\bm{c_1}
	+c_{1n}(\alpha(1-\bm{\beta^\mathrm{T}}\bm{e_{n-1}})) \\
	&+ (1-\alpha) \bm{\beta^\mathrm{T}}\bm{c_2} +c_{2n}((1-\alpha)(1-\bm{\beta^\mathrm{T}}\bm{e_{n-1}})).\\
	\end{aligned}
\end{equation}

The gradients (Eq.~\ref{partialvrvc}) can be written as
\begin{equation}
	\begin{aligned}
	&\partial_\alpha V_r(\alpha,\bm{\beta})=\frac{\partial V_r(\alpha ,\bm{\beta})}{\partial \alpha}=\bm{\beta^\mathrm{T}}\bm{u_r}+b_r, \\
	&\bm{\partial_{\beta}} V_c(\alpha,\bm{\beta})=\frac{\partial V_c(\alpha ,\bm{\beta})}{\partial \bm{\beta}}=\alpha\bm{u_c} +\bm{b_c},
	\end{aligned} 
\end{equation}
\noindent where $b_r=r_{1n}-r_{2n}$, $\bm{b_c}=\bm{c_2}-c_{2n}{\bm{e}}_{n-1}$, $\bm{u_r} =\bm{r_1}-\bm{r_2}-b_r\bm{e}_{n-1}$, and $\bm{u_c}=\bm{c_1}-\bm{c_2}-(c_{1n}-c_{2n})\bm{e}_{n-1}$.

\begin{theo}
	\label{thepart}
	If, in a 2-agent, $2\times n$, norm-form game, if there exists a $\delta>0$ such that the payoff matrices obey
	\begin{equation}\label{defpart}
		\bm{u_r}+\delta\bm{u_c}=0,
	\end{equation}
	and both agents follow the GA-SPP algorithm (with Condition 1, 2, and 3), then their strategies will converge to a Nash equilibrium. 
\end{theo}
\begin{proof}
    For a 2-agent $2\times n$ game, if we put Eq.~\ref{pVrVc} into Definition~\ref{defsemi0}, then we derive $\bm{u_r}+\bm{u_c}=0$. It shows that $2\times n$ games in Theorem~\ref{thepart} with $\delta=1$ are PSD games.
    
    First we consider $2\times n$ positive semi-definite games ($\bm{u_{r}}+\bm{u_{c}}=0$).
    According to Theorem \ref{thesemi}, in this particular case, GA-SPP can converge to a Nash Equilibrium. It means the following iteration can converge:  
\begin{equation}\label{partalg}
	\begin{aligned}
	&\overline{\alpha}_{k+1}=\Pi_{\Delta_1}[\alpha_k+ \gamma_k ( \bm{\beta_k}^\mathrm{T}\bm{u_{r_1}}+b_{r_1})],\\
	&\bm{\overline{\beta}_{k+1}}=\Pi_{\bm{\Delta_2}}[\bm{\beta_k}-  \gamma_k (\alpha_k \bm{u_{r_1}}+  \bm{b_{c_1}})];\\
	\ \ \  	&\alpha_{k+1}=\Pi_{\Delta_1}[\alpha_k+ \eta ( \bm{\overline{\beta}_k}^\mathrm{T}\bm{u_{r_1}}+b_{r_1})], \\
	&\bm{\beta_{k+1}}=\Pi_{\bm{\Delta_2}}[\bm{\beta_k}- \eta (\overline{\alpha_k} \bm{u_{r_1}}+  \bm{b_{c_1}})].
	\end{aligned}
\end{equation}
For brevity, we omit step 3 of GA-SPP.
    
For a $2\times n$, norm-form game that obeys Eq.~\ref{defpart}, we have $\bm{u_{r_2}}+\delta\bm{u_{c_2}}=0$. Let $x=\frac{\alpha}{\sqrt{\delta}},\ \bm{y}=\sqrt{\delta}\bm{\beta}$.
    If $\alpha$ and $\bm{\beta}$ follows GA-SPP, then the update rule of $x$ and $\bm{y}$ is
\begin{equation}\label{partalg2}
     \begin{aligned}
     &\overline{x}_{k+1}=\Pi_{\Delta_x}[x_k+ \gamma_k (\bm{y_k}^\mathrm{T}\bm{u_{r_2}} +\frac{b_{r_2}}{\sqrt{\delta}})],\\
     &\bm{\overline{y}_{k+1}}=\Pi_{\bm{\Delta_y}}[\bm{y_k}- \gamma_k (\alpha_k \bm{u_{r_2}}+  \sqrt{\delta}\bm{b_{c_2}})];\\
     &x_{k+1}=\Pi_{\Delta_x}[x_k+ \eta(\bm{\overline{y}_k}^\mathrm{T} \bm{u_{r_2}} +\frac{b_{r_2}}{\sqrt{\delta}})], \\
     &\bm{y_{k+1}}=\Pi_{\bm{\Delta_y}}[\bm{y_k}- \eta (\overline{x_k} \bm{u_{r_2}}+  \sqrt{\delta}\bm{b_{c_2}})].
     \end{aligned}
\end{equation}
\indent Comparing Eq.~\ref{partalg2} with Eq.~\ref{partalg}, $(x, \bm{y})$ can be viewed as a strategy pair of another $2\times n$ PSD game following GA-SPP. Notice that the proof of Theorem \ref{thesemi} only requires that the valid space is a bounded convex set. Therefore, if $(\alpha,\bm{\beta})$ follows GA-SPP, $(x,\bm{y})$ can converge, then $(x,\ \bm{y})$ can still converge in $2\times n$, norm-form game can converge. 
     
With Proposition~\ref{prop1}, we finish the proof of Theorem~\ref{thepart}.
\end{proof}

\begin{theo}
	\label{thepart1}
	If, in a 2-agent, $2\times n$, norm-form game, if there exists $\delta>0$, and the payoff matrices obey
	\begin{equation}\label{defpart}
		\bm{u_r}+\delta\bm{u_c}=0,
	\end{equation}
	and one agent follow the GA-SPP algorithm (with Condition 1, 2, and 3), another agent uses GA, then their strategies will converge to a Nash equilibrium. 
\end{theo}
    The proof of this theorem is omitted, which is similar to that of Theorem \ref{thepart}.
     

%% file: 5-3-Two-Actions.tex
\subsection{$2\times 2$ General-Sum Games}
In this section, we will prove the Nash convergence of GA-SPP in $2\times 2$ general-sum games.
\begin{theo}\label{theorem:22_general_sum_game}
	If, in a 2-agent, $2\times 2$, iterated general-sum game, both agents follow the GA-SPP algorithm (with  Condition 1, 2, and 3), then their strategies will converge to a Nash equilibrium. 
\end{theo}
\begin{proof}
With Proposition~\ref{prop1}, in order to prove Theorem~\ref{theorem:22_general_sum_game}, we just need to prove the convergence of GA-SPP in $2\times2$ games, which is accomplished by Lemma~\ref{lemma:22_0},~\ref{lemma:22_1}, and~\ref{lemma:22_2}. 
\end{proof}

Next, we will analyze the structure of $2\times2$ games firstly, and then show the convergence in different cases respectively.
 
In a 2-agent, 2-action game, the reward functions (Eq.~\ref{equa:VrVc}) can be written as
\begin{equation}\nonumber\begin{aligned}
\label{vrvc2}
V_r(\alpha,\ \beta) & = r_{11}(\alpha \beta) + r_{12}(\alpha (1 - \beta))  + r_{21}((1 - \alpha) \beta)\\
& + \ r_{22}((1 - \alpha)(1 - \beta)),\\
V_c(\alpha,\ \beta) & = c_{11}(\alpha \beta) + c_{12}(\alpha (1 - \beta))  + c_{21}((1 - \alpha) \beta)\\
& + \ c_{22}((1 - \alpha)(1 - \beta)).
\end{aligned}\end{equation}
And the gradient function (Eq.~\ref{partialvrvc}) can be written as

\begin{equation}\nonumber\begin{aligned}
\label{partialvrvc2}
\partial_{\alpha} V_r(\alpha,\ \beta) = \frac{\partial V_r(\alpha,\ \beta)}{\partial \alpha} = u_r\beta + b_r,\\
\partial_{\beta} V_c(\alpha,\ \beta) = \frac{\partial V_c(\alpha,\ \beta)}{\partial \beta} = u_c\alpha + b_c,
\end{aligned}\end{equation}

\noindent where $u_r = r_{11} + r_{22} - r_{12} - r_{21}$,~$b_r = r_{12} - r_{22}$,~$u_c = c_{11} + c_{22} - c_{12} - c_{21}$, and~$b_c = c_{21} - c_{22}$. We have $|u_r|\leq2\delta_r,\ |u_c|\leq2\delta_c$.

 We can formulate the first two update rules of GA-SPP (\ref{alg:gaspp}):
\begin{equation}\label{gaspp2}
\begin{aligned}
	\alpha_{k+1}=\Pi_\Delta [\alpha_k+\eta\partial_{\alpha_k} V_r(\alpha_k,\ \Pi_\Delta[\beta_k+\gamma\partial_{\beta_k},\ \beta_k])],\\
	\beta_{k+1}=\Pi_\Delta [\beta_k+\ \eta\partial_{\beta_k} V_c(\beta_k,\ \Pi_\Delta[\alpha_k+\gamma\partial_{\alpha_k},\ \alpha_k])],
\end{aligned}
\end{equation}
  where $\Delta=\Delta_1=\Delta_2=[0,\ 1]$.
%

To prove the Nash convergence of GA-SPP, we will examine the dynamics of the strategy pair following GA-SPP. In a 2-agent, 2-action, general-sum game, $(\alpha,\ \beta)$ can be viewed as a point in $\mathbb{R}^2$ constrained to lie in the unit space.

According to Eq.~\ref{gaspp2}, if $(\alpha_k,\ \beta_k)$ is a unconstrained point, then value of $(\alpha_{k+1},\ \beta_{k+1})$ is
\begin{equation}\begin{aligned}\label{uncon}
	\begin{bmatrix}
	\alpha_{k+1}\\ \beta_{k+1}
	\end{bmatrix}\ -\ 
	\begin{bmatrix}
	\alpha_{k}\\ \beta_{k}
	\end{bmatrix}\
	= \eta
	&\begin{bmatrix}
	\gamma_k u_r u_c & u_r\\ u_c & \gamma_k u_r u_c
	\end{bmatrix}
	\begin{bmatrix}
	\alpha_{k}\\ \beta_{k}
	\end{bmatrix} \\
	+ \eta
	&\begin{bmatrix}
	\gamma_k u_r b_c+b_r\\ \gamma_k u_c b_r+b_c
	\end{bmatrix}.
\end{aligned}\end{equation}

We denote the $2\times 2$ matrix in Eq.~\ref{uncon} as $\mathrm{U}$. If the matrix $\mathrm{U}$ is invertible, in the unconstrained condition, there exists and only exists one point so that the left hand side of Eq.~\ref{uncon} is zero. We call this point the center (or origin) and denote it as $(\alpha^c,\ \beta^c)$. The eigenvalues of $\mathrm{U}$ is given by
\begin{equation}\begin{aligned}\label{eigenvalue}
\lambda_1=\gamma_k u_r u_c+\sqrt{u_r u_c}\ \ \text{and} \ \ \lambda_2=\gamma_k u_r u_c-\sqrt{u_r u_c}.
\end{aligned}\end{equation}
According to Condition 2 ($4\gamma_k^2\delta_r\delta_c<1$) and $|u_r|\leq2\delta_r,\ |u_c|\leq2\delta_c$, then $\gamma_k^2u_r u_c<1$.
There are three cases of $\mathrm{U}$:
\begin{itemize}
	\item{Case 1:} $u_r u_c=0$, \ie, $\mathrm{U}$ is not invertible;
	\item{Case 2:} $u_r u_c<0$, \ie, having two imaginary conjugate eigenvalues with negative real;
	\item{Case 3:} $u_r u_c>0$, \ie, having two real eigenvalues.
\end{itemize}
To prove Theorem~\ref{theorem:22_general_sum_game}, we only need to show that GA-SPP always leads the strategy pair to converge in these three cases.

\begin{figure}[t]
\setlength{\abovecaptionskip}{0.2cm}
\setlength{\belowcaptionskip}{-0.5cm}
	\includegraphics[width=8.8cm, height=4.4cm]{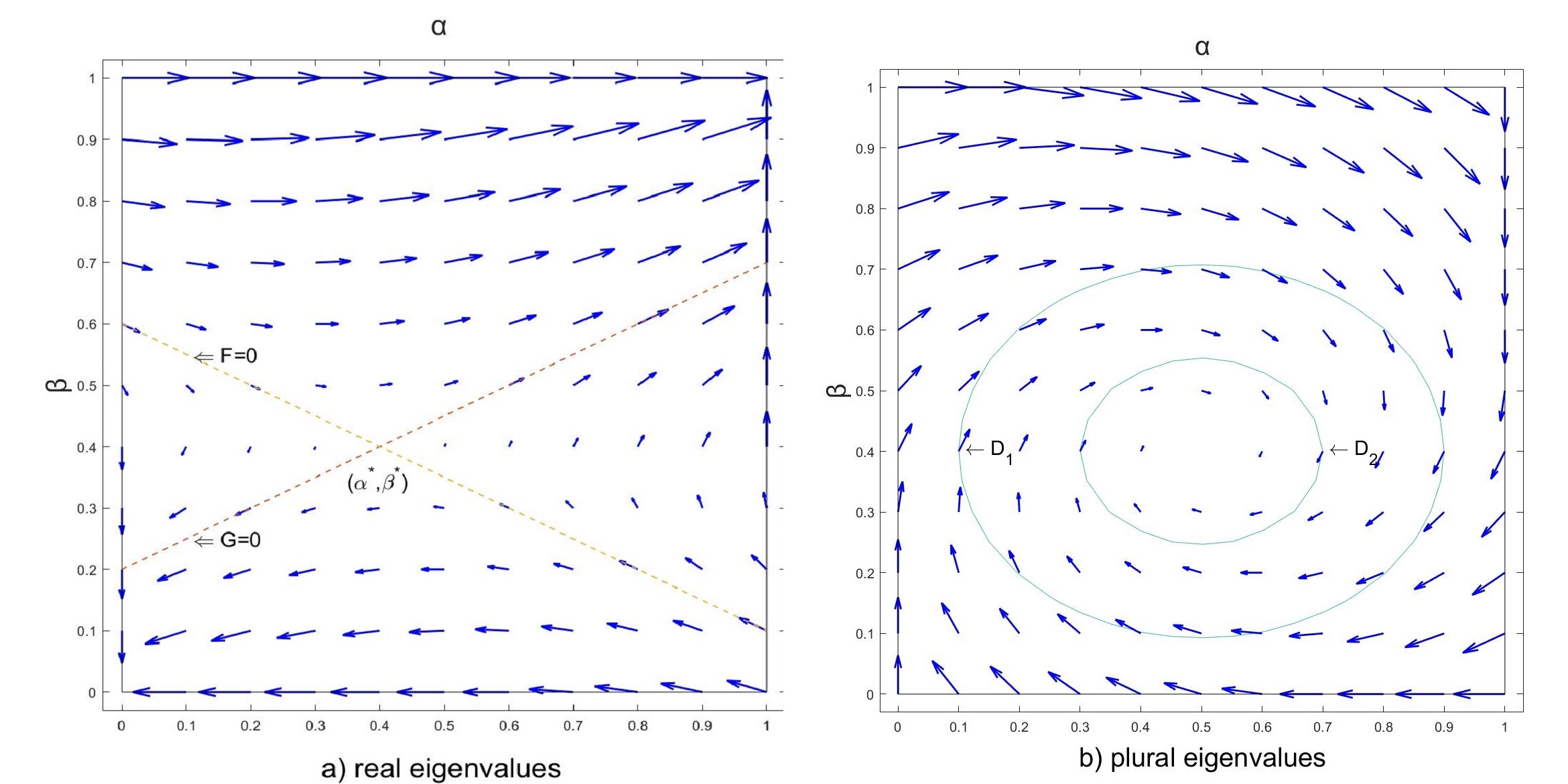}
	\caption{Strategy updating directions of the GA-SPP. a) when $\mathrm{U}$ has real eigenvalues and b) when $\mathrm{U}$ has imaginary eigenvalues with negative real part.}
\end{figure}

\begin{figure*}[t]
\setlength{\abovecaptionskip}{0cm}
\setlength{\belowcaptionskip}{-0.5cm}
    \subfigure[]{
	    \includegraphics[width=0.3\linewidth]{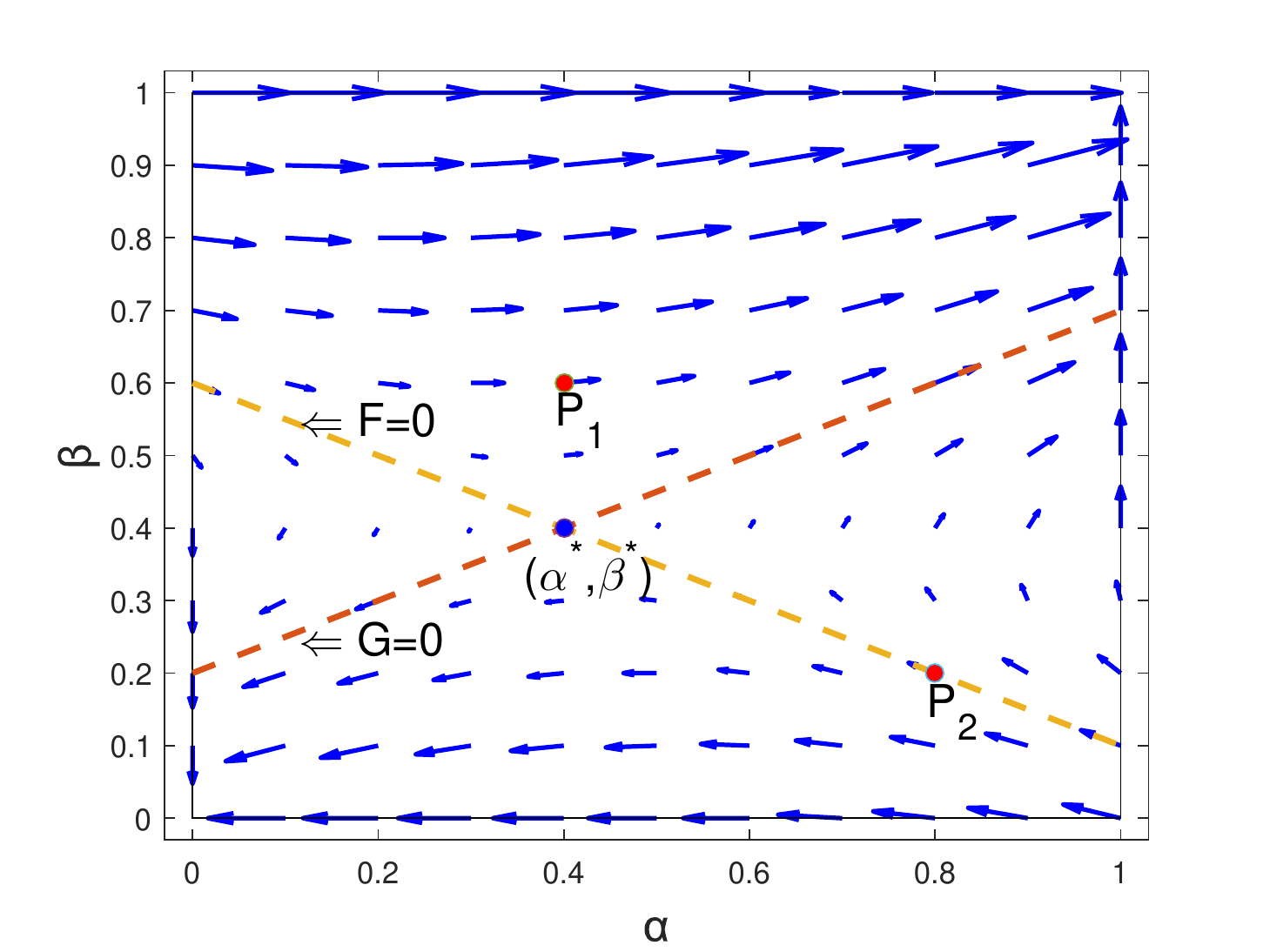}\label{fig:4_3_0}
	}\hfill
	\subfigure[]{
	    \includegraphics[width=0.3\linewidth]{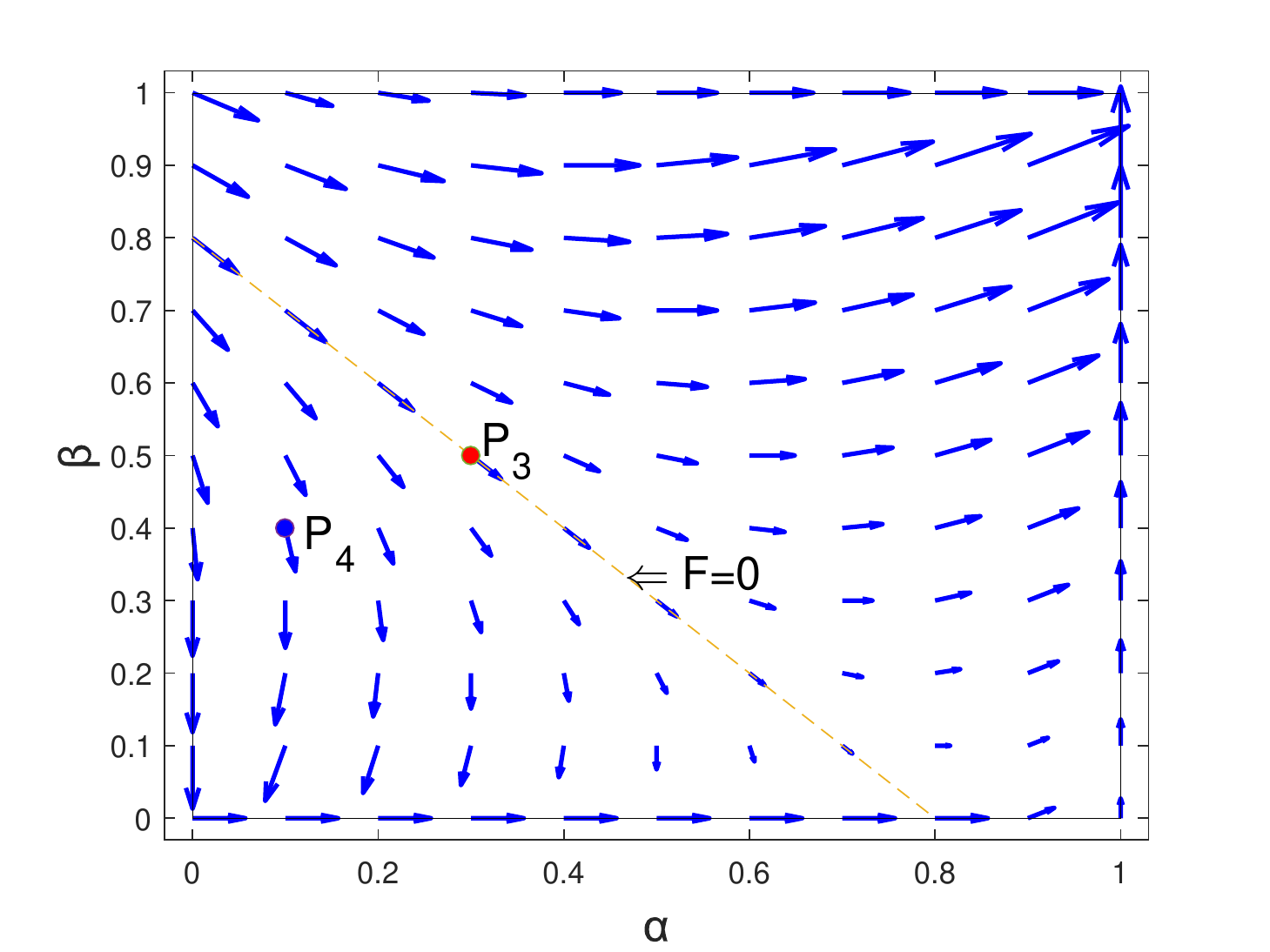}\label{fig:4_3_1}
	}\hfill
	\subfigure[]{
	    \includegraphics[width=0.3\linewidth]{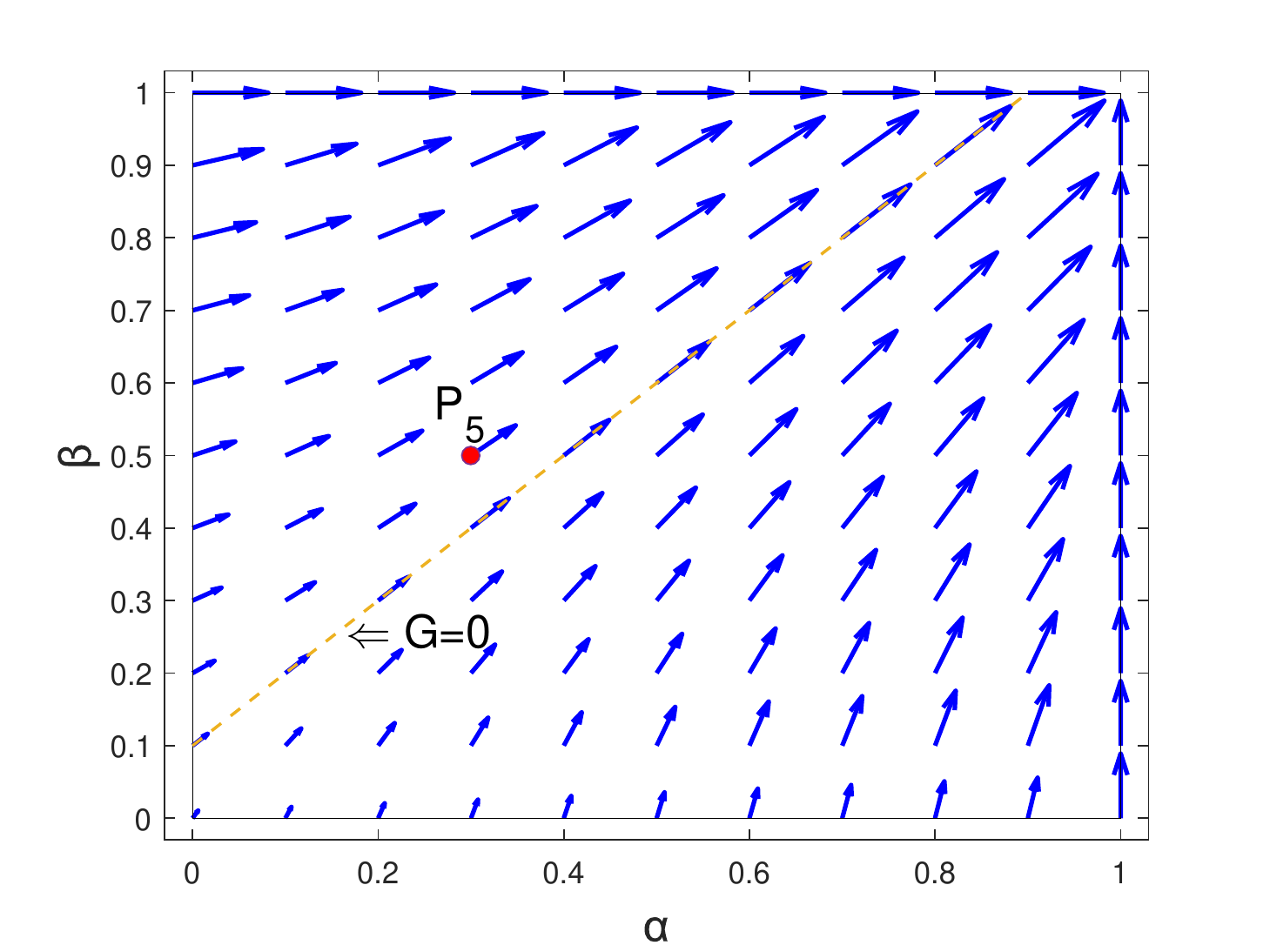}\label{fig:4_3_2}
	}
	\caption{Updating directions of strategy pair in Lemma~\ref{lemma:22_2}. (a) Both~$\alpha^c$ and~$\beta^c$ are in the valid probability range, $F(P_1)>0$, $F(P_2)=0$; (b) only one of $\alpha^c$ and $\beta^c$ is in the valid range, $F(P_3)=0$, $F(P_4)<0$; (c) neither $\alpha^c$ and $\beta^c$ is in the valid range.}
\end{figure*}

\begin{lemma}\label{lemma:22_0}
	If, in a 2-agent, 2-action, iterated general-sum game, $\mathrm{U}$ is not invertible, for any initial strategy pair, GA-SPP leads the strategy pair trajectory to converge to a Nash equilibrium (NE) with finite steps.	
\end{lemma}
\begin{proof} 
	From Eq.~\ref{eigenvalue}, if $\mathrm{U}$ is not invertible, then $u_r u_c=0$. Assume $u_c=0$  (the analysis for the case $u_r=0$ is analogous and thus omitted for brevity). 
	
	According to Eq.~\ref{uncon}, $\beta_{k+1}=\Pi_\Delta[\beta_k+\eta b_c]$. Because $\eta b_c$ is  constant and $\beta\in [0,\ 1]$, strategy $\beta$ will no longer change after finite steps. We denote this value by $\beta^*$. Then, $\alpha_{k+1}=\Pi_\Delta[\alpha_k+\eta(u_r\beta^*+b_r)]$. Because $\eta(u_r\beta^*+b_r)$ is a constant and $\alpha\in [0,\ 1]$, after a certain number (finite) of steps, strategy $\alpha$ will also stop changing. We denote this value by $\alpha^*$.
	So in this case, the strategy pair will converge to $(\alpha^*,\ \beta^*)$, and with Proposition~\ref{prop1}, this is a Nash equilibrium.
	
	Note that the index of $\eta_k$ was omitted in the proof, which is because the situation of $(\alpha_{k+1},\ \beta_{k+1})=(\alpha_k,\ \beta_k)\ \& \ (\overline{\alpha}_{k+1},\ \overline{\beta}_{k+1})\ne (\alpha_k,\ \beta_k) $ did not occur in this case. Lemma~\ref{lemma:22_2} also has this property.
\end{proof}

\begin{lemma}\label{lemma:22_1}
	If, in a 2-agent, 2-action, iterated general-sum game, $\mathrm{U}$ has two imaginary conjugate eigenvalues with negative real, for any initial strategy pair, GA-SPP leads the strategy pair trajectory to converge to a NE.
\end{lemma}
\begin{proof}
	Since $u_r u_c<0$, there exists a $\delta$ such that $u_r+\delta u_c=0$. This is a 2-dimensional situation of Theorem~\ref{thepart}. So GA-SPP can converge to a NE.
\end{proof} 

In the rest of this section, we will introduce Lemma~\ref{lemma:22_2} and the basic idea of proof. For the detailed mathematical proof, we refer readers to the supplementary material. 
\begin{lemma}\label{lemma:22_2}
	If, in a 2-agent, 2-action, iterated general-sum game, $\mathrm{U}$ has real eigenvalues, for any initial strategy pair, GA-SPP leads the strategy pair trajectory to converge to a point that is a NE.
\end{lemma}

Before proof, we first introduce some variables to simplify the expressions.

If $\mathrm{U}$ is invertible, then $u_r u_c\ne0$. Let $x=\alpha+\frac{b_c}{u_c}$, $y=\beta+\frac{b_r}{u_r}$, Eq.~\ref{uncon} can be reformulated as:
\begin{equation}
    \begin{aligned}
		\label{eq:XY}
		\begin{bmatrix}
		x_{k+1}\\ y_{k+1}
		\end{bmatrix}-
		\begin{bmatrix}
		x_{k}\\ y_{k}
		\end{bmatrix}\ = \eta
		\begin{bmatrix}
		\gamma u_r u_c & u_r\\ u_c & \gamma u_r u_c
		\end{bmatrix}
		\begin{bmatrix}
		x_{k}\\ y_{k}
		\end{bmatrix}
    \end{aligned}
\end{equation}
By setting the left hand side of Eq.~\ref{eq:XY} to zero, we can get an equation, the only solution of which is $x=0$, $y=0$.

Now we give the proof of Lemma~\ref{lemma:22_2}.
\begin{proof} 
	From Eq.~\ref{eigenvalue}, real eigenvalues imply $u_r u_c>0$. So without the loss of generality, we assume that $u_r > 0$ and $u_c > 0$ (the analysis for the case $u_r<0$ and $u_c < 0$ is analogous and thus omitted).
	
	Proof of Lemma~\ref{lemma:22_2} depends on the location of $(\alpha^c,\ \beta^c)$, which has three possibilities:
	\begin{enumerate}
		\item both $\alpha^c$ and $\beta^c$ are in the valid probability range [0,1], 
        \item only one of $\alpha^c,\ \beta^c$ is in the valid probability range [0,1],
        \item neither  $\alpha^c$ nor $\beta^c$ is in the valid probability range [0,1].
	\end{enumerate}
	Proofs of convergence in these three cases are given in Property \ref{property:22_3_0}, \ref{property:22_3_1}, and \ref{property:22_3_2}, respectively.
\end{proof}

Notice that $\mathrm{U}$ has two real eigenvalues: $\lambda_1>0$ and $\lambda_2<0$ and two nonparallel eigenvectors. The central point $(x=0,\ y=0)$ with two eigenvectors $(\bm{v_1}=[\sqrt{u_r},\ \sqrt{u_c}],\ \bm{v_2}=[\sqrt{u_r},\ -\sqrt{u_c}])$ can form a new 2D coordinate system. The basic idea of proof is to analyze coordinates of the strategy pair update trajectory in the new coordinate system. To be brief, we introduce two functions to compute it instead of converting coordinate.
\begin{equation}\begin{aligned}
	F=x+\sqrt{\frac{u_r}{u_c}}y,\ \ \ G=x-\sqrt{\frac{u_r}{u_c}}y.
\end{aligned}\end{equation}

\begin{Property}\label{property:22_3_0}
If~$\ \mathrm{U}$ has real eigenvalues, both of $\alpha^c$, $\beta^c$ are in the valid probability range([0,\ 1]), GA-SPP leads the strategy pair trajectory to converge to a NE.
\end{Property}
\begin{proof}
As shown in Fig.~\ref{fig:4_3_0}, the initial point will affect the Nash convergence result because there are three Nash equilibrium points. 

The first case is~$F_0=0$. Then the strategy pair point will keep staying in the line $F=0$ while the absolute value of $G$ decreases, which means that the point moves to the center point~$(F=0,G=0)$, \ie, $(\alpha^*,\ \beta^*)$. For example, if $P_2$ is the initial point, the point will travel along the line~$F=0$ and moves to~$(F=0,G=0)$. We can compute $F_{k+1}$ and $G_{k+1}$ by
\begin{equation}
    \begin{aligned}
        &G_{k+1}=(1+\eta\lambda_2) G_k,\ \ \ &F_{k+1}=(1+\eta\lambda_1) F_k.
    \end{aligned}
\end{equation}   
\indent According to Condition 1, 2 and 3, $0<(1+\eta\lambda_2) <1$, so the GA-SPP will converge to $(F=0,\ G=0)$, \ie, $(x=0,\ y=0)$.

Another case is when~$F_0>0$, from Fig.~\ref{fig:4_3_0}, we can tell that the strategy pair point first touches the boundary $x_{max}\ (\alpha=1)$ or $y_{max}\ (\beta=1)$ after finite iteration steps, after then it travels along the boundary and moves to $(x_{max},\ y_{max})$. For example, if~$P_1$ is the initial point, the point will touch the boundary $\alpha=1$ (\ie, $x_{max}$), then it travels along $\alpha=1$ and moves to $(x_{max},\ y_{max})$.
Without no more than one exceptional case, we can derive $F_{k+1}>F_{k}$ in each iteration. From the monotone bounded theorem, the GA-SPP will converge to $F_{max}$, \ie, $(x_{max},\ y_{max})$.

Situation when~$F_0<0$ is similar to that when $F_0>0$, so we omit it for brevity.

In all, GA-SPP can converge for any initial strategy pair.
 \end{proof}
\begin{Property}\label{property:22_3_1}
When~$\mathrm{U}$ has real eigenvalues and only one of $\alpha^c$ and $\beta^c$ is in the valid probability range, GA-SPP leads the strategy pair trajectory to converge to a NE.
\end{Property}
Proof of Property~\ref{property:22_3_1} can be classified into 4 cases. Without loss of generality, we just consider one of them, where $\beta^c<0$ and $\alpha^c \in [0,1]$. As shown in Fig.~\ref{fig:4_3_1}, we can also divide the proof into 3 parts: $F_0>0$, $F_0<0$, and $F_0=0$. If the point is in the part $F>0$, according to Property~\ref{property:22_3_0}, the algorithm will converge to $(x_{max},\ y_{max})$. If $F_0\leq0$, we can see that the point will first touch a boundary of the valid probability space, after that it will move into the part $F>0$. For example, if $P_3\ (F=0)$ is the initial point, the point will travel along line $F=0$ until it hits the boundary, then it will be projected to the subspace where $F>0$. If $P_4\ (F<0)$ is the initial point, it will touch the boundary $\beta=0\ (y_{min})$ and then $y$ remains $y_{min}$ while $x$ increases until it move into the subspace where $F>0$.

\begin{Property}\label{property:22_3_2}
If~$\ \mathrm{U}$ has real eigenvalues and neither $\alpha^c$ nor $\beta^c$ is in the valid probability range [0, 1], GA-SPP leads the strategy pair trajectory to converge to a NE.
\end{Property}
Updating directions of strategy pair is shown in Fig.~\ref{fig:4_3_2} for this case. For the detailed proof, please refer to supplementary material.

\begin{theo}\label{theorem:22_general_sum_game1}
	If, in a 2-agent, $2\times 2$, iterated general-sum game, one agent follow the GA-SPP algorithm (with  Condition 1, 2, and 3), another agent uses GA, then their strategies will converge to a NE
\end{theo}
    The proof is omitted, which is similar to that of Theorem \ref{theorem:22_general_sum_game}.

%% file: 6-Experiments.tex
\section{Experimental Analysis in Normal-Form Games}
In this section, we will illustrate GA-SPP in games with experiments and compare GA-SPP with IGA-PP and GIGA-WoLF, both of which have theoretical guarantees, in some larger games.

\subsection{Benchmark games}
We first illustrate the results of GA-SPP on four representative benchmark games presented in Tab.~\ref{tab:benchmark_games}. GA-SPP converges to NE in all of these games (Fig.~\ref{fig:benchmark_games}).

\begin{table}[h]
\setlength{\abovecaptionskip}{-0.1cm}
\setlength{\belowcaptionskip}{-0.2cm}
\caption{Benchmark games}\label{tab:benchmark_games}
\subtable[Prisoners' Dilemma]{
    \begin{tabular}{|p{0.96cm}|p{0.96cm}|p{0.96cm}|}
	    \hline
	    &Silent&Betray\\
	    \hline
        Silent&(-1,-1)& (-3,0)\\
	    \hline
        Betray&(0,-3)& (-2,-2)\\
	    \hline
    \end{tabular}\label{tab:prisoners' dilemma}
}\hfill
\subtable[Chicken]{
\begin{tabular}{|p{0.96cm}|p{0.96cm}|p{0.96cm}|}
	\hline
	&Swerve&Straight\\
	\hline
	Swerve&(-2,-2)& (1,-1)\\
	\hline
	Straight&(-1,1)& (-1,-1)\\
	\hline
\end{tabular}\label{tab:chicken}
}
\subtable[Battle of Sexes]{
\begin{tabular}{|p{0.96cm}|p{0.96cm}|p{0.96cm}|}
	\hline
	& Opera & Football\\
	\hline
    Opera&(3,2)& (1,1)\\
	\hline
    Football&(0,0)& (2,3)\\
	\hline
\end{tabular}\label{tab:battles of sexes}
}\hfill
\subtable[Rock-Paper-Scissors]{
\begin{footnotesize}
\begin{tabular}{|p{0.64cm}|p{0.64cm}|p{0.64cm}|p{0.64cm}|}
	    \hline
	    &R&P&S\\
	    \hline
        R&(0,0)&(-1,1)&(1,-1)\\
	    \hline
        P&(1,-1)&(0,0)&(-1,1)\\
        \hline
        S&(-1,1)&(1,-1)&(0,0)\\
	    \hline
\end{tabular}
\end{footnotesize}\label{tab:rock-paper-scissors}
}

\end{table}

\begin{figure}
\setlength{\abovecaptionskip}{0.1cm}
\setlength{\belowcaptionskip}{-0.5cm}
    \includegraphics[width=0.48\linewidth]{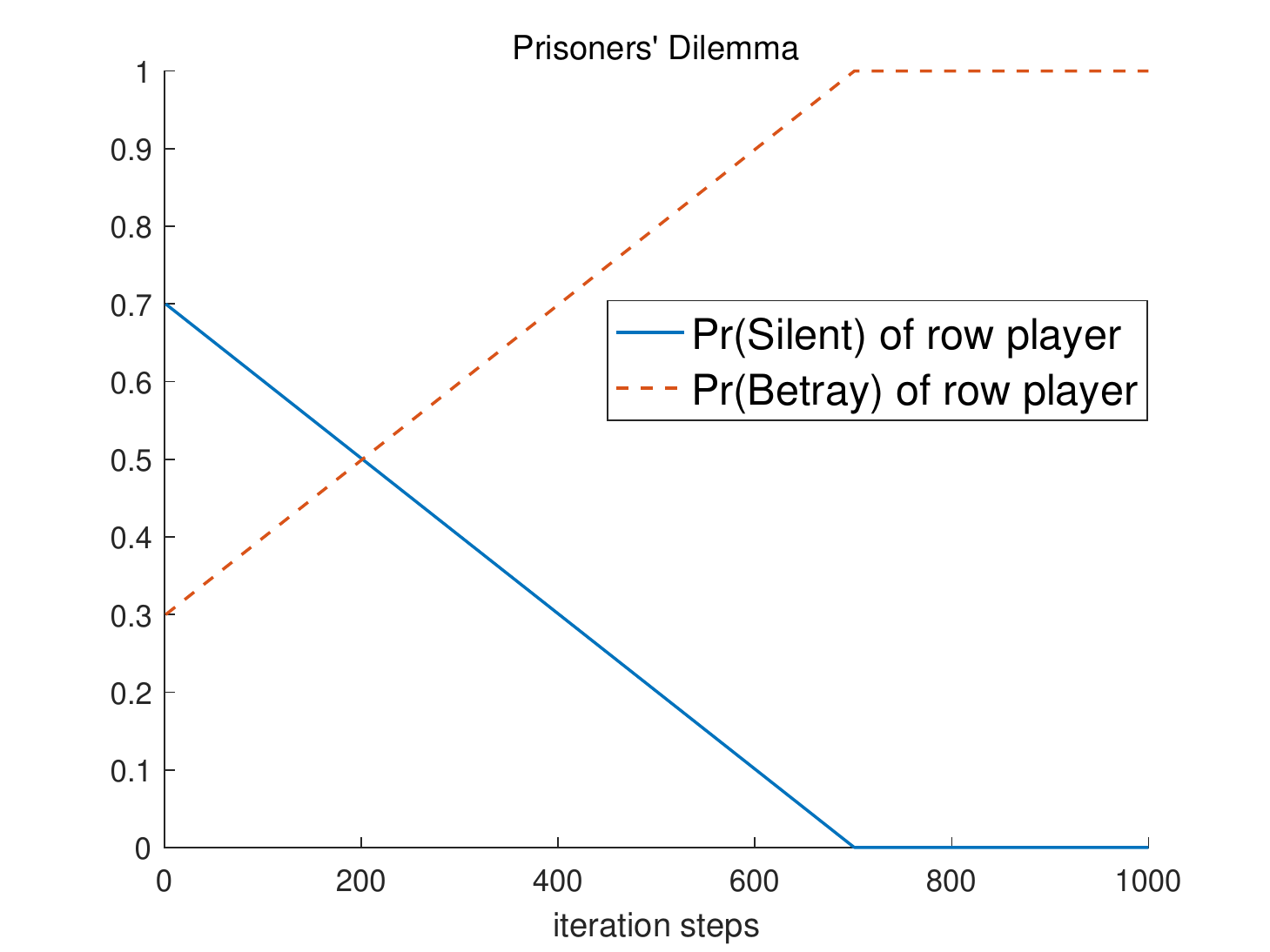}
    \includegraphics[width=0.48\linewidth]{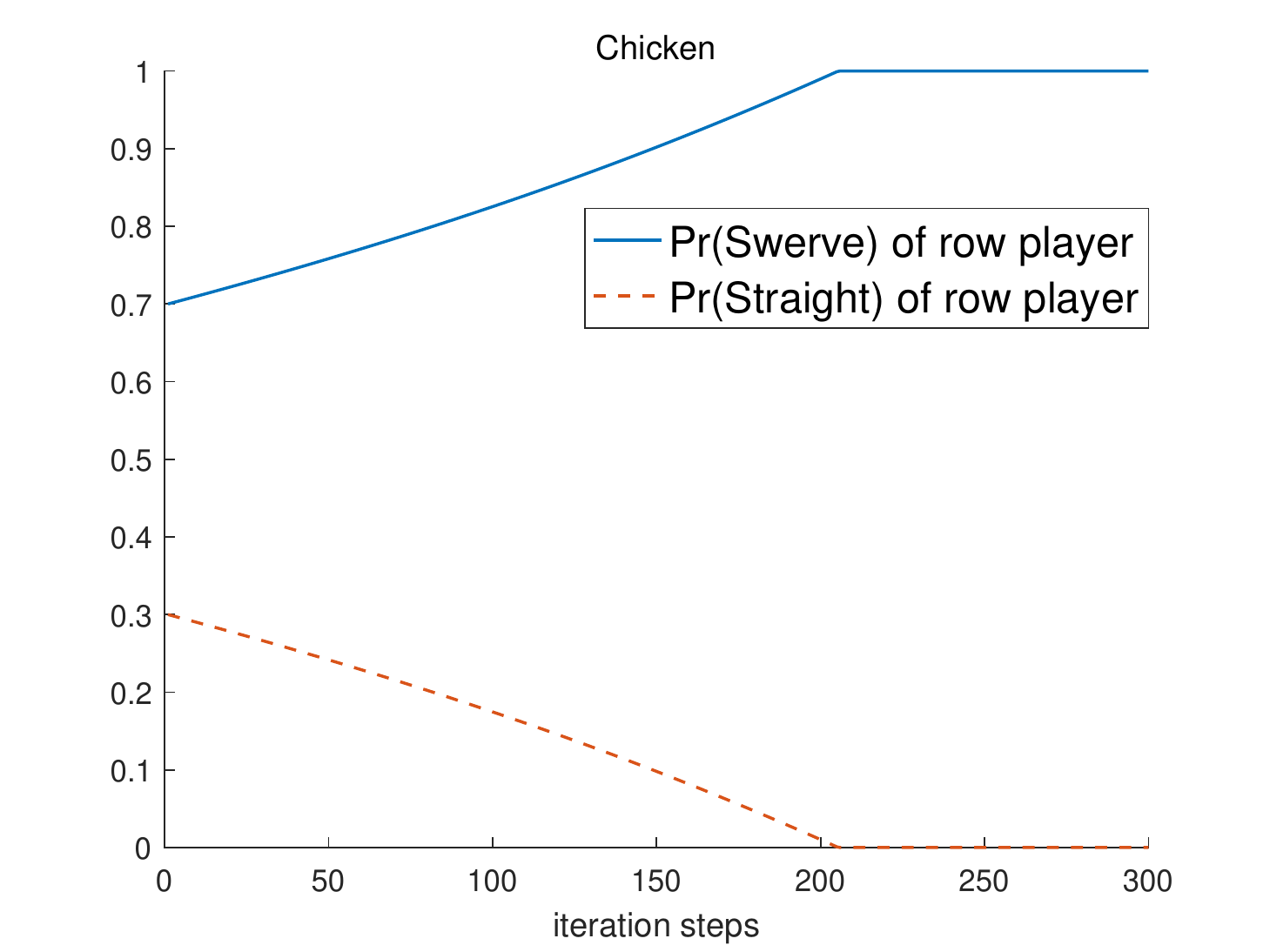}
    \includegraphics[width=0.48\linewidth]{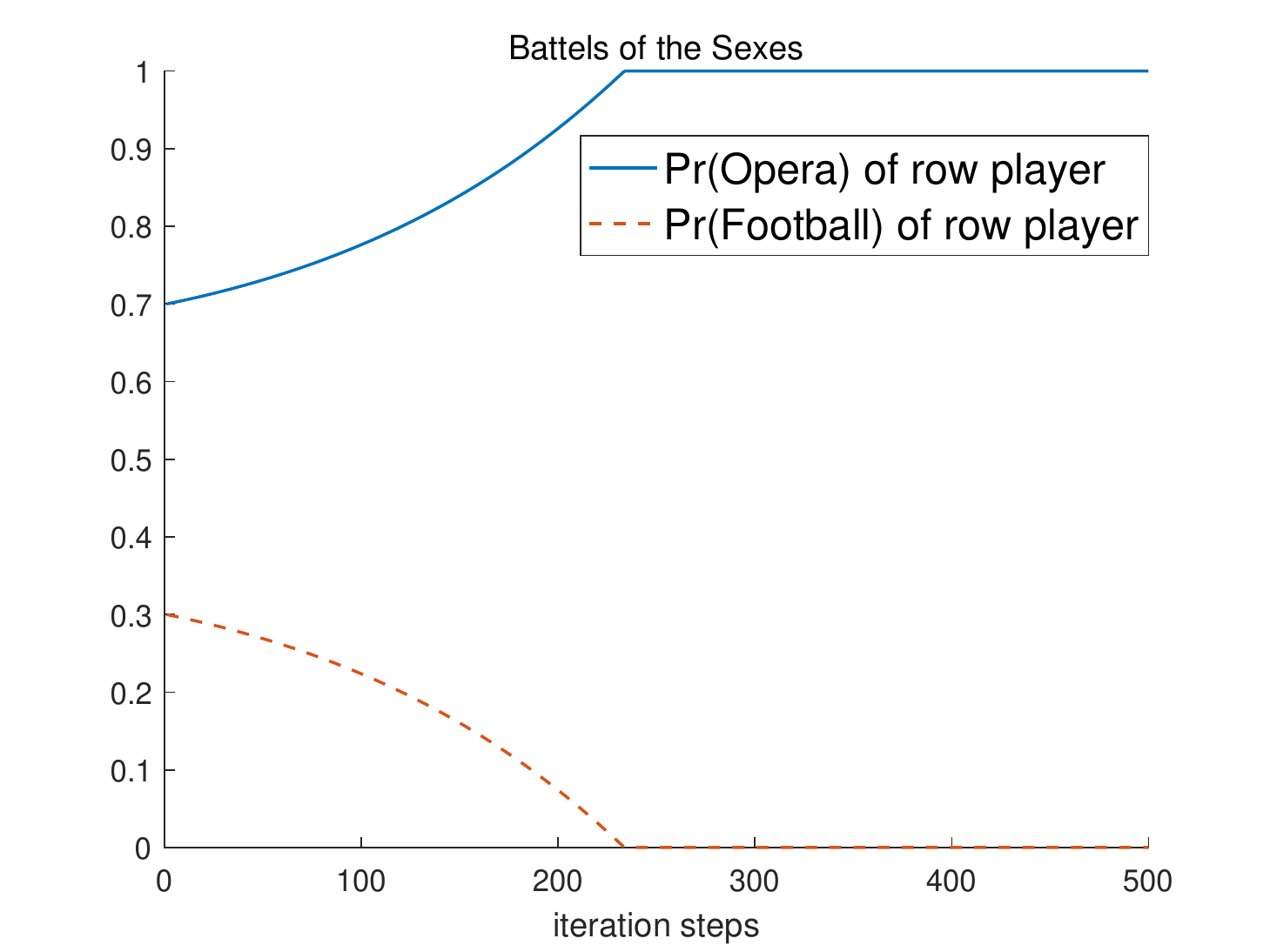}
    \includegraphics[width=0.48\linewidth]{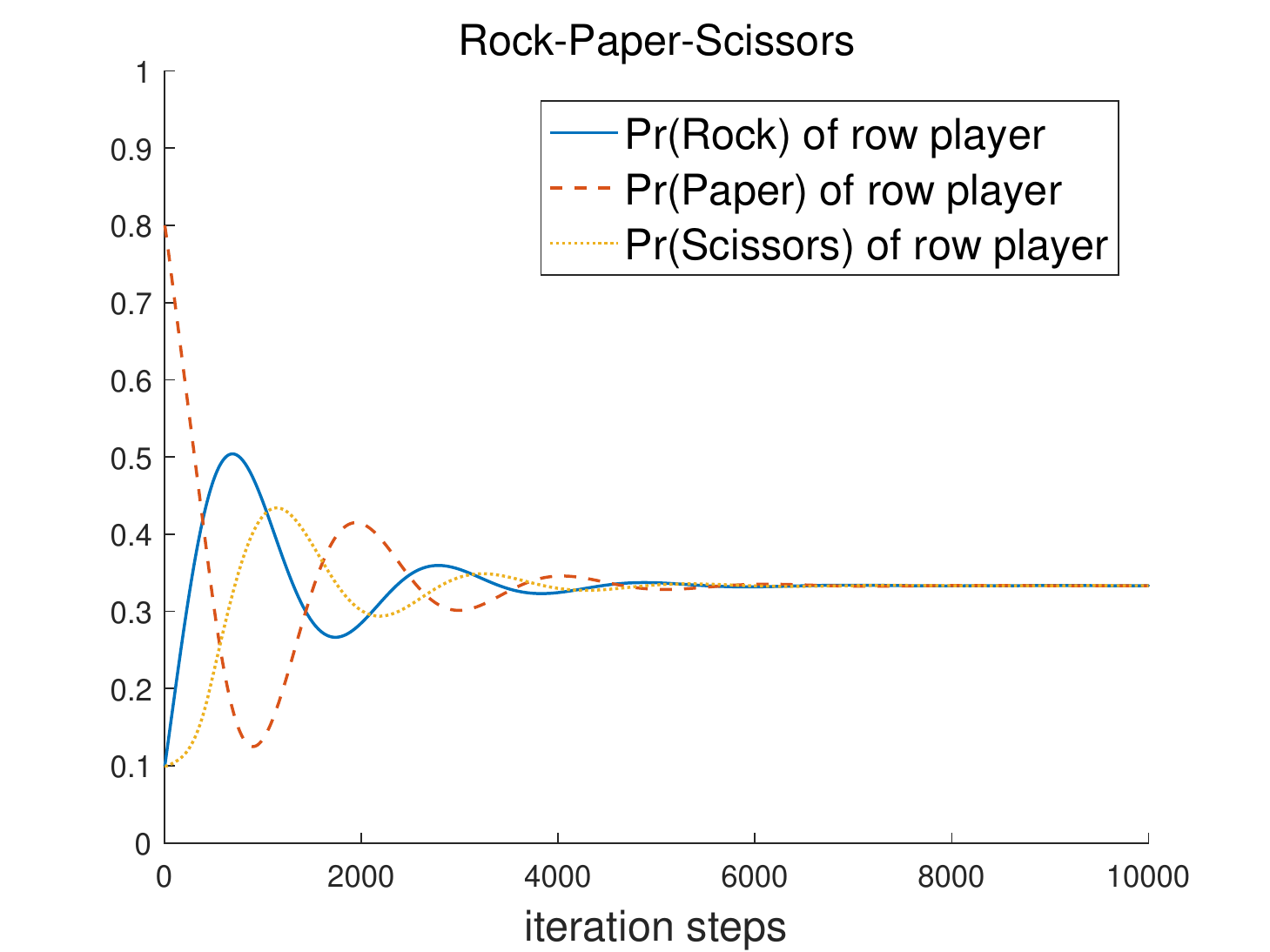}
    \caption{Action probabilities of row agent following GA-SPP in four benchmark games. Parameters: $\eta=0.001$, $\gamma=0.1$. Initial polices: $(0.7,\ 0.3)$ and $(0.3,\ 0.7)$.}
    \label{fig:benchmark_games}
\end{figure}

\subsection{ Games beyond theoretical settings}
We also evaluate GA-SPP in \emph{Shapley's game} and a $2\times 3$ game, presented in Tab.~\ref{tab:larger_games}. Although the theoretical analyses of GA-SPP have not covered these games, empirical results show that it still converge. We now compare GA-SPP, GIGA-WoLF, and IGA-PP in these two games. 

Fig.~\ref{fig:larger_games} shows the row player's action probabilities over time if both players follow GA-SPP, GIGA-WoLF, or IGA-PP in \emph{Shapley's game} respectively. GIGA-WoLF fails to converge in this non-zero sum game, but GA-SPP and IGA-PP can converge to a Nash equilibrium.

\begin{table}[h]
\setlength{\abovecaptionskip}{-0.1cm}
\caption{Games with larger settings}\label{tab:larger_games}
\subtable[Shapley's Game]{
    \begin{tabular}{|p{0.65cm}|p{0.65cm}|p{0.65cm}|p{0.65cm}|}
	    \hline
	    &C1&C2&C3\\
	    \hline
        R1&(0,0)&(1,0)&(0,1)\\
	    \hline
        R2&(0,1)&(0,0)&(1,0)\\
        \hline
        R3&(1,0)&(0,1)&(0,0)\\
	    \hline
    \end{tabular}\label{tab:shapley's_game}
}\hfill
\subtable[A 2x3 Game]{
    \begin{tabular}{|p{0.65cm}|p{0.65cm}|p{0.65cm}|p{0.65cm}|}
	    \hline
	    &C1&C2&C3\\
	    \hline
        R1&(3,3)&(0,5)&(1,-2)\\
	    \hline
        R2&(2,2)&(1,1)&(-1,0)\\
        \hline
    \end{tabular}\label{tab:a_2_3_game}
}

\end{table}

\begin{figure*}
\setlength{\abovecaptionskip}{0.2cm}
\setlength{\belowcaptionskip}{0cm}
    \includegraphics[width=0.3\linewidth]{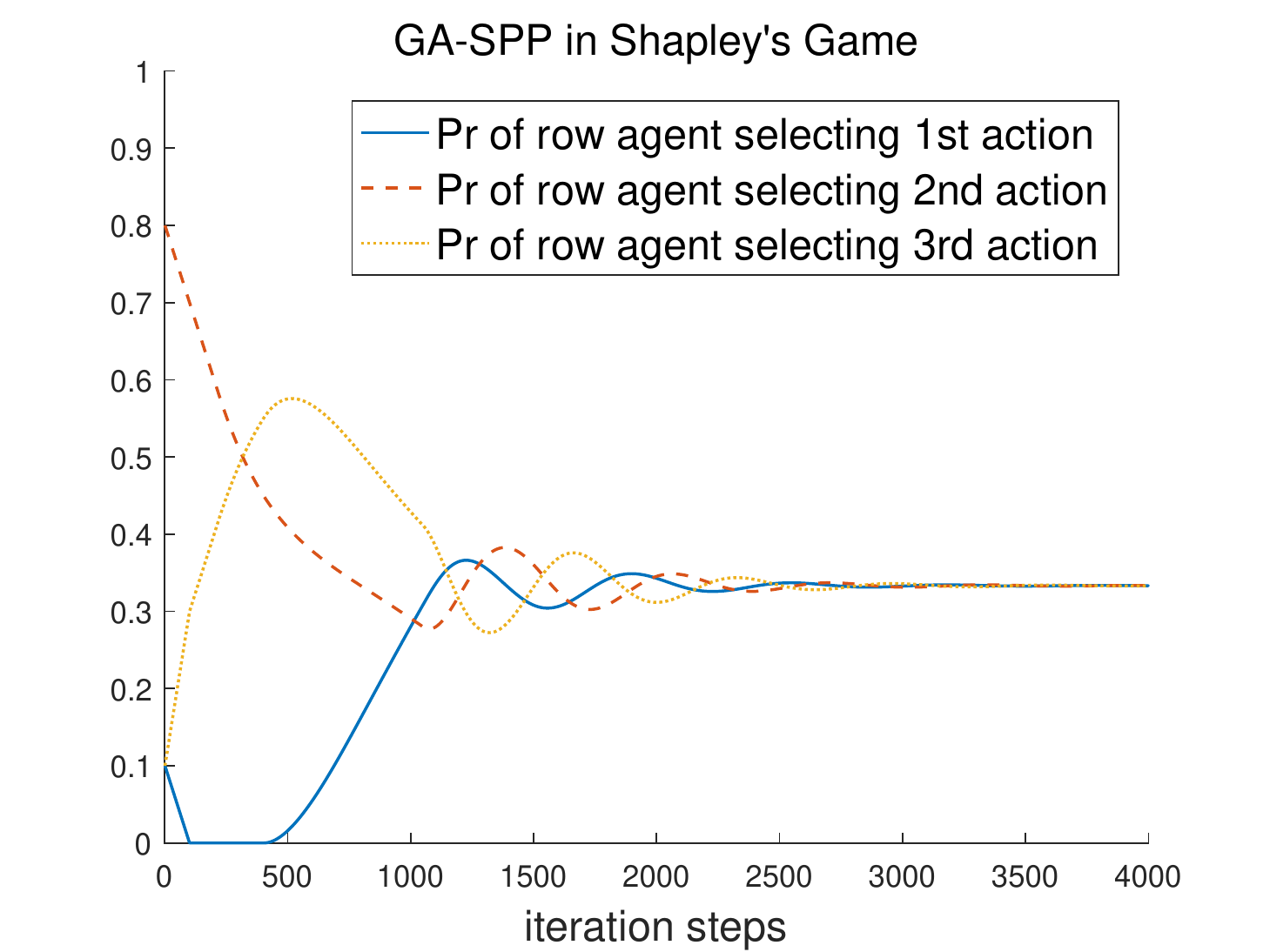}
    \includegraphics[width=0.3\linewidth]{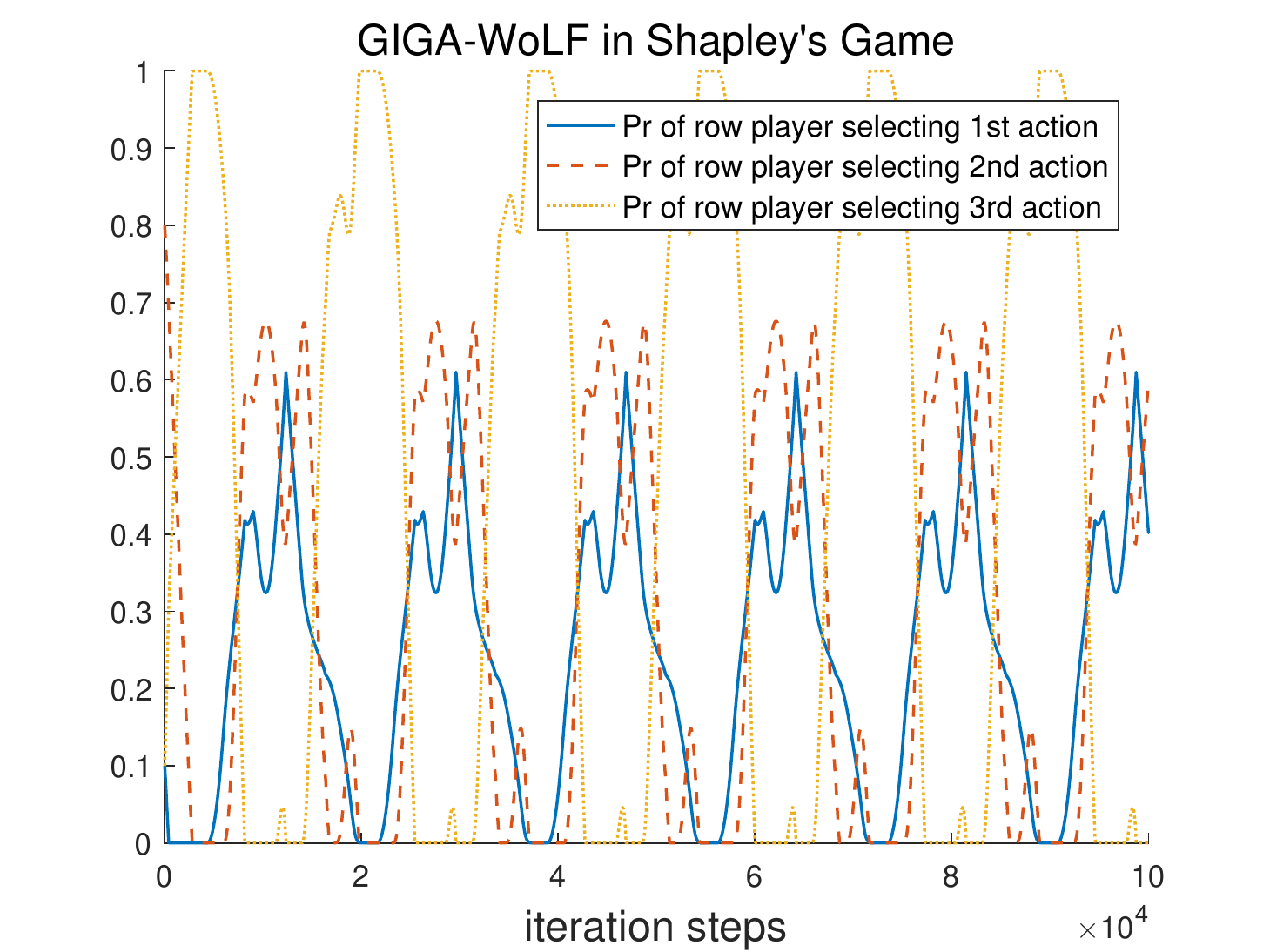}
    \includegraphics[width=0.3\linewidth]{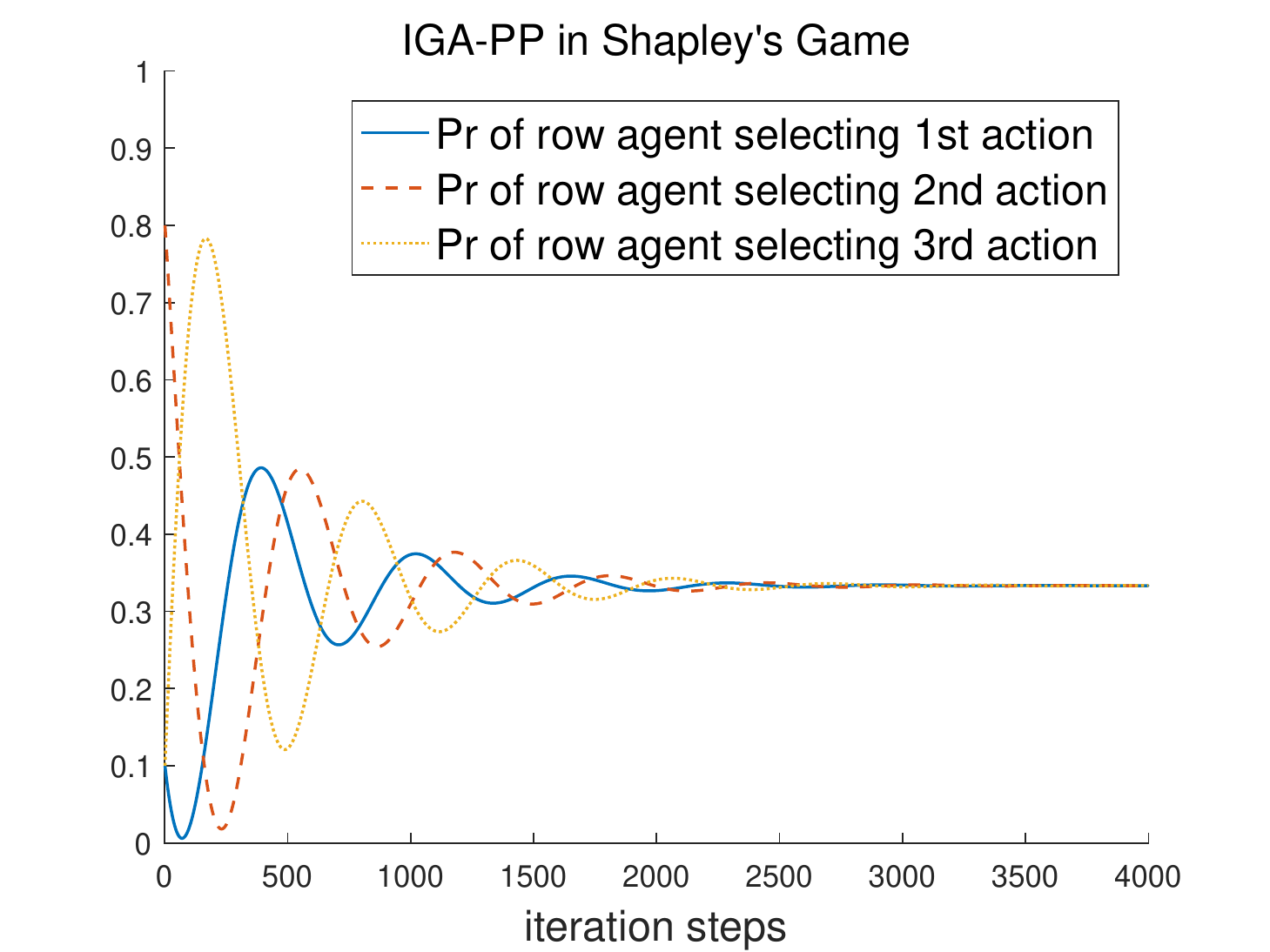}
    \caption{Comparison between GA-SPP, GIGA-WoLF, and IGA-PP in \emph{Shapley's game}. GIGA-WoLF cannot converge while GA-SPP and IGA-PP converge to NE. GA-SPP has a slighter oscillation. Parameters: $\gamma=3,\ \eta=0.001$. Initial polices: $(0.1,\ 0.8,\ 0.1)$ and $(0.8,\ 0.1,\ 0.1)$.}
    \label{fig:larger_games}
\end{figure*}

\begin{figure}
\setlength{\abovecaptionskip}{0.2cm}
\setlength{\belowcaptionskip}{-0.5cm}
    \includegraphics[width=0.48\linewidth]{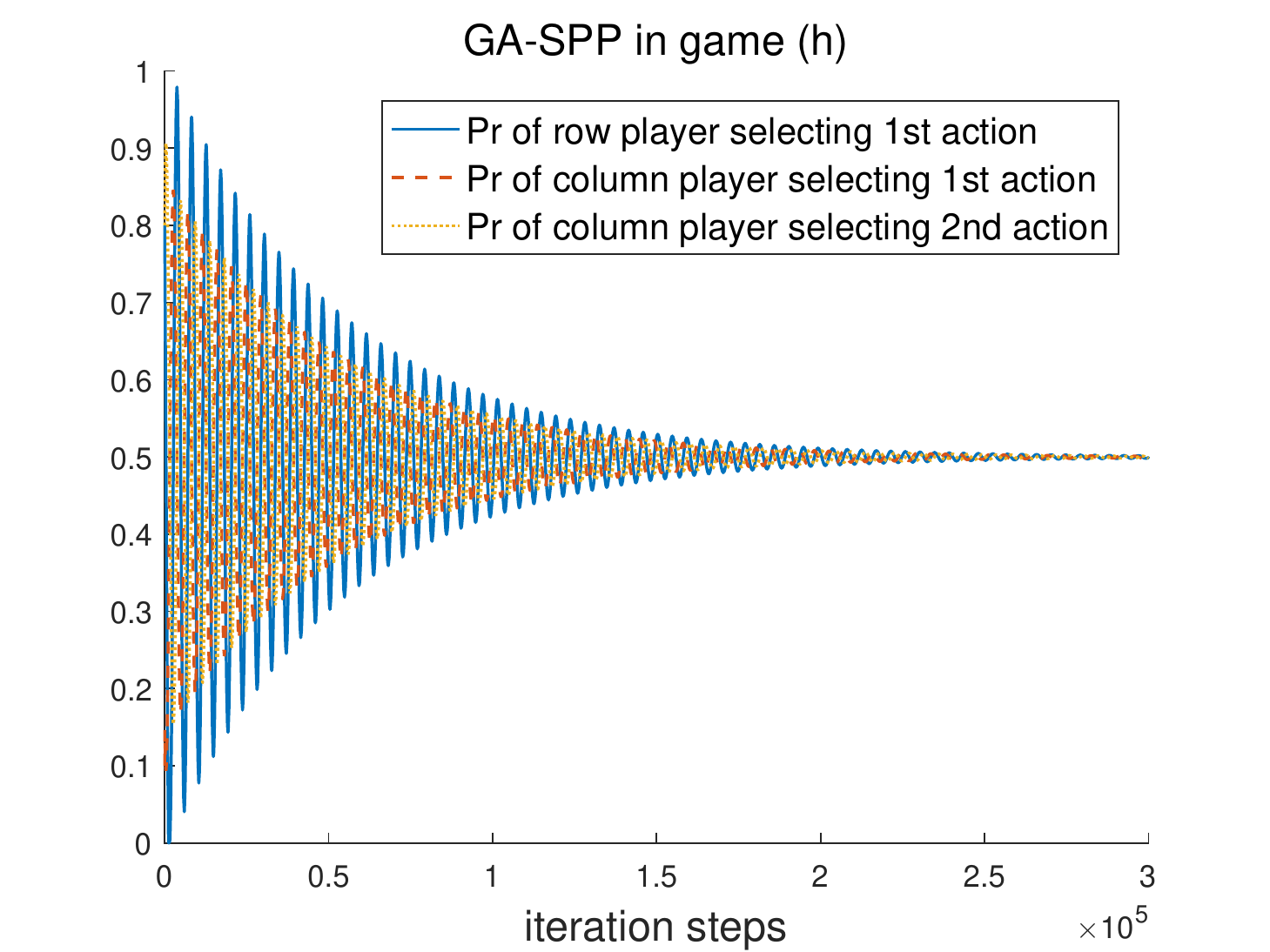}
    \includegraphics[width=0.48\linewidth]{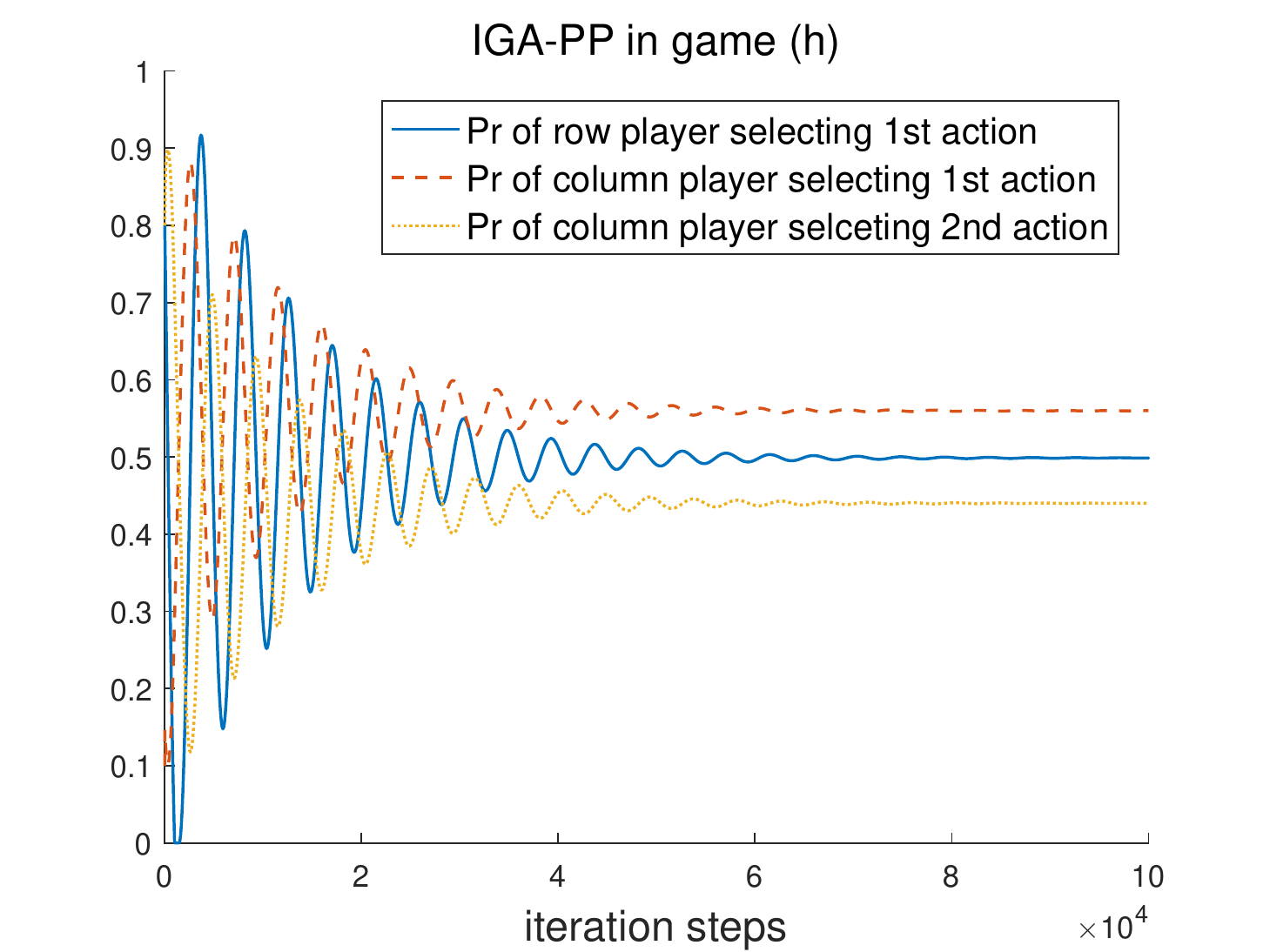}
    \includegraphics[width=0.48\linewidth]{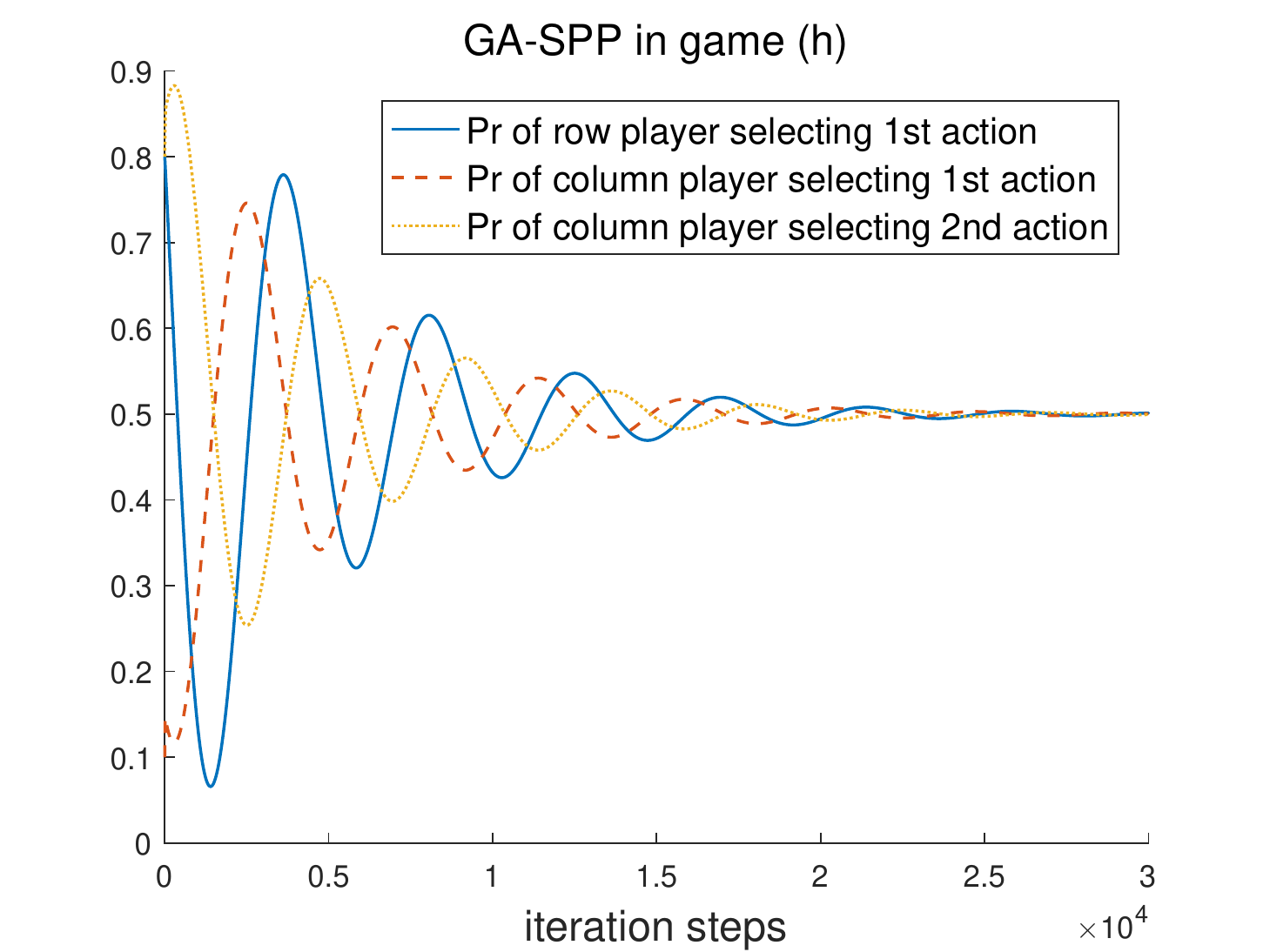}
    \includegraphics[width=0.48\linewidth]{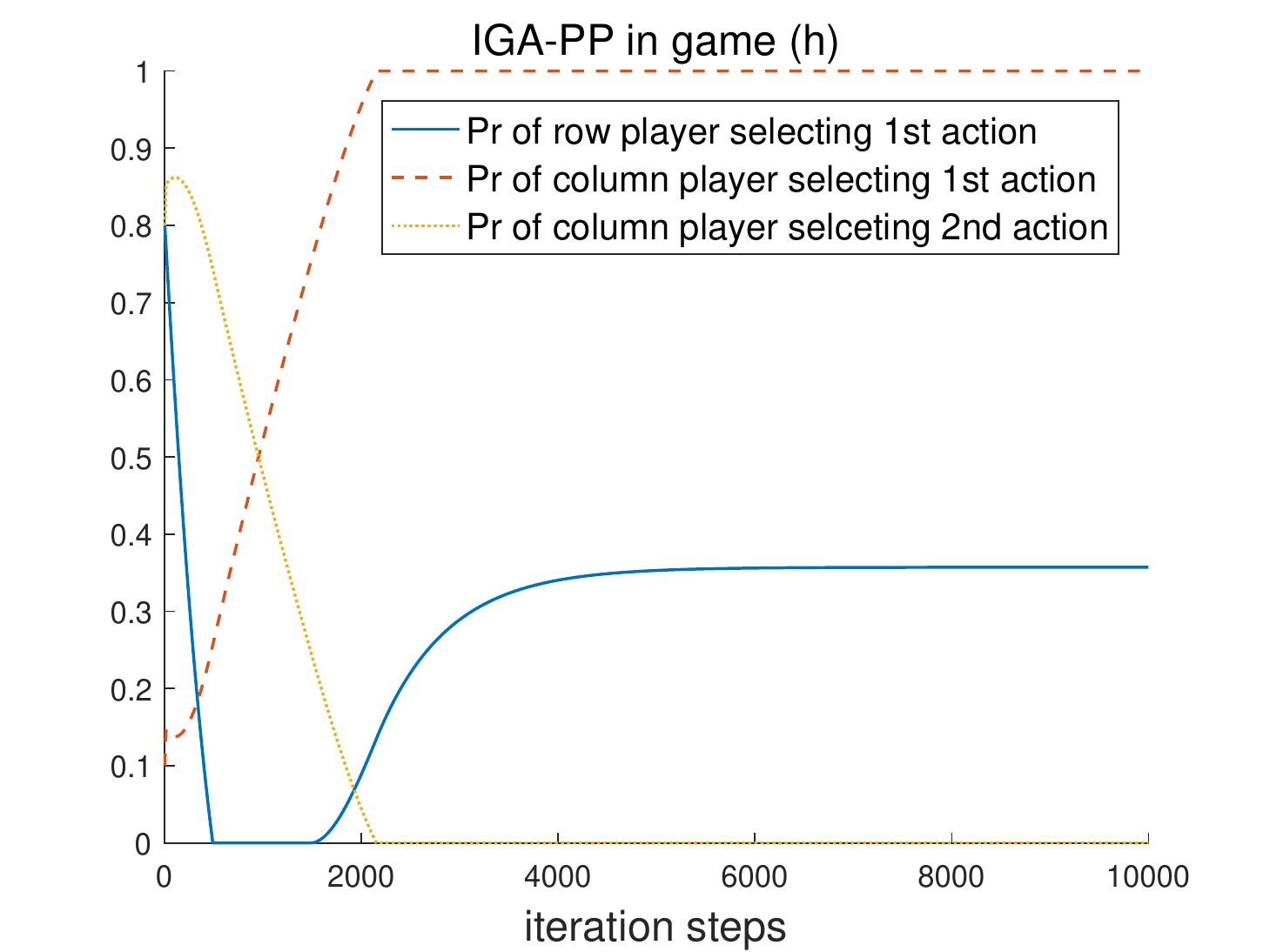}
    \caption{Comparison between GA-SPP, and IGA-PP in a $2\times 3$ game under different prediction lengths. IGA-PP's convergence to Nash Equilibrium is affected by prediction length, while GA-SPP can always converge to NE. Parameters: $\eta=0.001$, $\gamma=0.01$ in upper and  $\gamma=0.1$ in lower. Initial polices: $(0.8,\ 0.2)$ and $(0.1,\ 0.8,\ 0.1)$.}
    \label{fig:a_2_3_game}
\end{figure}

Fig.~\ref{fig:a_2_3_game} shows results of GA-SPP and IGA-PP in a $2\times 3$ game under different prediction lengths. Although IGA-PP can converge, it does not converge to a Nash equilibrium. On the contrary, the strategies lead by GA-SPP successfully converge to Nash equilibrium under different prediction lengths. The essential reason is that GA-SPP projects the predicted strategies to a valid space at every step.

By examining with different learning rates, we observe that GA-SPP often converges faster than GIGA-WoLF. A possible explanation is introduced in~\cite{zhang2010multi}. We do not show these results for sake of space.

\subsection{Problem games}
Although GA-SPP has better performance than other MAL algorithms, the convergence of GA-SPP is not perfect. As shown in Fig.~\ref{fig:three2_game}, in the \emph{three player matching pennies}, GA-SPP cannot converge with a constant prediction length.

\begin{figure}
\setlength{\abovecaptionskip}{0.2cm}
\setlength{\belowcaptionskip}{-0.5cm}
    \includegraphics[width=0.7\linewidth]{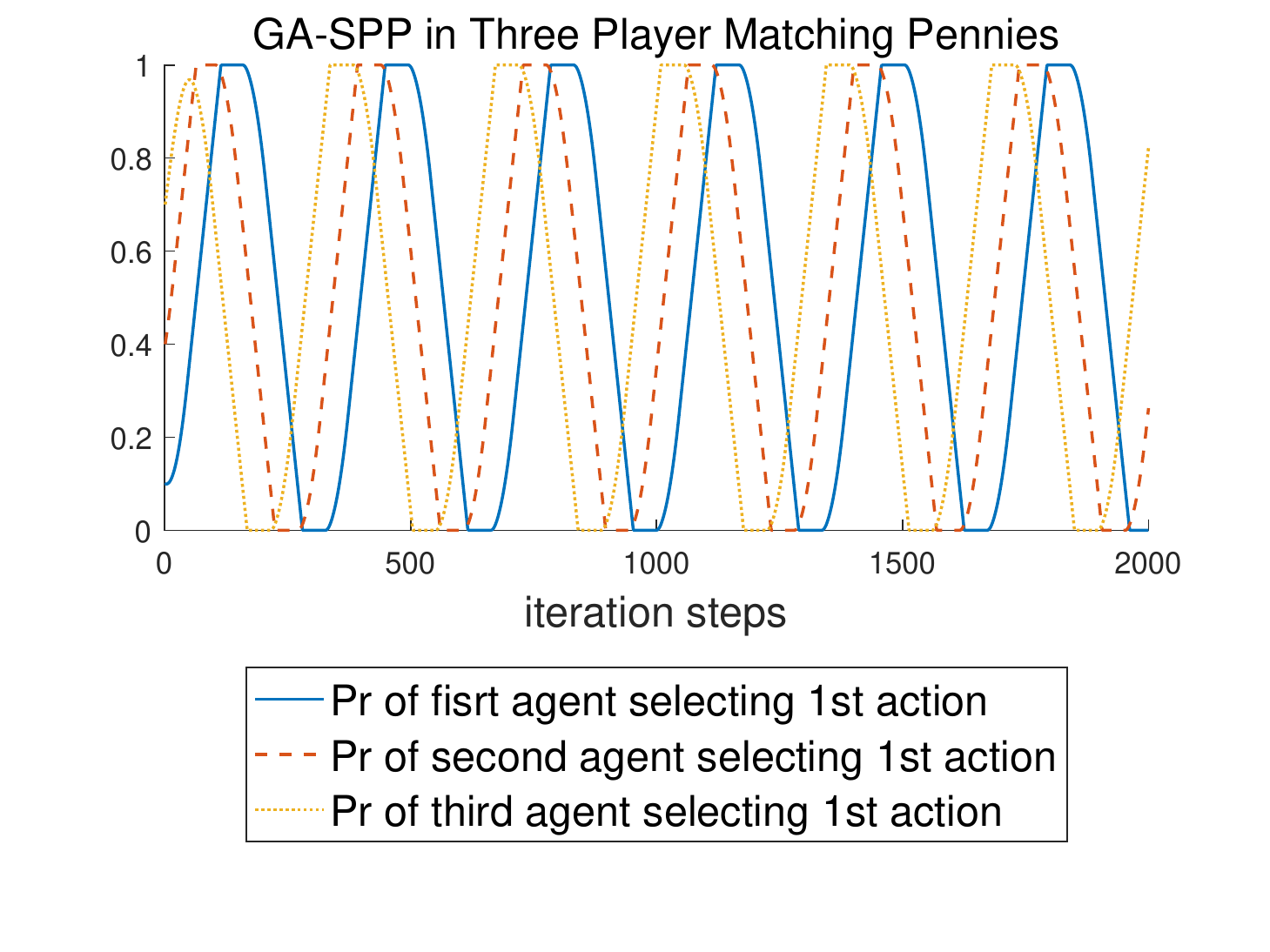}
    \caption{\footnotesize{Following GA-SPP, action probabilities of agents fail to converge in three player matching pennies. Parameters: $\eta=0.001$,  $\gamma=0.3$. Initial polices: $(0.1,\ 0.9)$, $(0.4,\ 0.6)$ and $(0.7,\ 0.3)$.}}
    \label{fig:three2_game}
\end{figure}

This failed case show the difficulties of MAL work and indicate that gradient method may not be the ideal way to handle a complex game. Because dynamic of gradient method in such game is not linearly, the chaotic phenomenon will occur. We may need different approaches to deal with such problems. In order to make MARL work effectively in more cases, it is important to analyze and solve these problems.